\documentclass[a4paper,10pt]{article}
\usepackage{graphics} 
\usepackage{subfigure}
\usepackage{graphicx,color}
\usepackage[active]{srcltx}
\usepackage{latexsym}
\usepackage{times} 
\usepackage{amsmath} 
\usepackage{amssymb}  
\usepackage{amsmath,amssymb,amsfonts,theorem}
\usepackage{enumerate}
\DeclareMathOperator*{\bigcart}{\raisebox{-0.3ex}{\text{\Large$\times$}}}
{ \theorembodyfont{\normalfont} 

\def\stopmodif{\color{black}} 

\newcommand{\R}{\mathbb{R}}

\renewcommand{\S}{\mathcal{S}}    

\newcommand{\sign}{{\rm sign}}

\newtheorem{remark}{Remark}

\newtheorem{corollary}{Corollary}
\newtheorem{proposition}{Proposition}

\newcommand{\dst}{\displaystyle}
\def\be{\begin{equation}}
\def\ee{\end{equation}}
\def\ba{\begin{array}}
\def\ea{\end{array}}
\def\eqa{\begin{eqnarray}}
\def\eqe{\end{eqnarray}}

\title{\LARGE \bf
...*
}
\date{}
\title{Formation control with binary information}
\author{Matin Jafarian  and  Claudio De Persis
\thanks{M.~Jafarian and C.~De Persis are with ITM, Faculty of Mathematics and
Natural Sciences, University of Groningen, the
Netherlands, Tel: +31 50 363 3080.
        {\tt\small \{m.jafarian,c.de.persis\}@rug.nl}.
        This work is partially supported by 
the Dutch Organization
for Scientific Research (NWO) under the auspices of the project {\it QUantized Information Control for formation Keeping} (QUICK).}}
\begin{document}
\maketitle

\begin{abstract}
In this paper, we study the problem of formation keeping of a network of strictly passive systems 
when very coarse information is exchanged. We assume that neighboring agents only know whether their relative position is larger or smaller than the prescribed one. This assumption results in very simple control laws that direct the agents closer or away from each other and take values in finite sets. We show that the task of formation keeping while tracking a desired trajectory and rejecting matched disturbances is still achievable under the very coarse information scenario. In contrast with other results of practical convergence with coarse or quantized information, here the control task is achieved exactly.
\end{abstract}
\section{Introduction}
Distributed motion coordination of mobile agents has attracted increasing attention in recent years owing to its wide range of applications from biology and social networks to sensor/robotic networks. Distributed formation keeping control is a specific case of motion coordination which aims at reaching a desired geometrical shape for the position of the agents while 
tracking a desired velocity. This problem has been addressed with different approaches \cite{tanner.et.al.cdc03.I}, \cite{arcak.tac07}, \cite{ren.book}, \cite{bullo.et.al.book}, \cite{mesbahi.egerstedt.book}, \cite{bai.et.al.book}, \cite{ren.book.II}. In problems of formation control, an important component, besides the dynamics of the agents and the graph topology, is the flow of information among the agents. In fact, the usual assumption in the literature on cooperative control is that a continuous flow of perfect information is exchanged among the agents. However, due to the coarseness of sensors and/or to communication constraints, the latter might be a restrictive requirement. To cope with the problem in the case of continuous-time agents' dynamics, the use of distributed quantized feedback control has been  proposed in the literature \cite{cortes.aut06}, \cite{dimarogonas.johansson.aut10}, \cite{ceragioli.et.al.aut11}, \cite{chen.et.al.aut11}. In fact, in formation control by quantized feedback control, 
information is transmitted among the agents whenever measurements cross the thresholds of the quantizer. At these times, the corresponding quantized value taken from a discrete set is transmitted. This allows to deal naturally both with the continuous-time nature of the agents' dynamics and with the discrete nature of the transmission information process without the need to rely on sampled-data models \cite{ceragioli.et.al.aut11}, \cite{depersis.cdc11}, \cite{depersis.jaya}.\\ 
{\it Literature review.} The research on coordination in the presence of quantized or coarse information has mainly focused on discrete-time systems (\cite{kashyap.et.al.aut07}, \cite{frasca.et.al.ijrnc09}, \cite{carli.et.al.ijrnc10}, \cite{AN-AO-AO-JNT:09} to name a few). Motivated by the problem of reaching consensus in finite-time, \cite{cortes.aut06} adopted binary control laws and cast the problem in the context of discontinuous control systems. The paper \cite{chen.et.al.aut11} proposed a consensus control algorithm in which the information collected from the neighbors is in binary form. The paper \cite{dimarogonas.johansson.aut10} has investigated a similar problem but in the presence of quantized measurements, while \cite{ceragioli.et.al.aut11} has rigorously cast the problem in the framework of non-smooth control systems. The latter has also introduced a new class of hybrid quantizers to deal with possible chattering phenomena. Results on the formation control problem for continuous-time models of the agents in the presence of coarse information have also appeared. The authors of \cite{Yu.et.al.2012} study a rendezvous problem for Dubins cars based on a ternary feedback which depends on whether or not the neighbor's position is within the range of observation of the agent's sensor. Deployment on the line of a group of kinematic agents was shown to be achievable by distributed feedback which uses binary information control \cite{depersis.et.al.cdc11}. In \cite{depersis.cdc11}, \cite{depersis.jaya} formation control problems for groups of agents with second-order non-linear dynamics have been studied in the presence of quantized measurements. Leader-follower coordination for single-integrator agents and constant disturbance rejection by proportional-integral controllers with quantized information and time-varying topology has been studied in \cite{xargay.et.al.cdc12}. In this paper, we consider nonlinear agents and we adopt internal model based controllers to track  reference velocities and reject constant and harmonic disturbances. 
On the other hand, we restrict ourselves to static graph topologies. 
\\
{\it Main contribution.} In this paper, we study the problem of distributed position-based formation keeping of a group of agents with strictly-passive dynamics which exchange binary information. The binary information models a sensing scenario in which each agent detects whether or not the components of its current distance vector from a neighbor are above or below the prescribed distance and apply a force (in which each component takes a binary value) to reduce or respectively increase the actual distance. A similar coarse sensing scenario was considered in \cite{Yu.et.al.2012} in the context of the so-called ``minimalist" robotics. \\
 Remarkably, despite such a coarse information and control action, we show that the control law guarantees exact achievement of the desired formation. This is an interesting result, since statically quantized control inputs typically generates practical convergence, namely the achievement of an approximate formation in which the distance from the actual desired formation depends on the quantizer resolution \cite{depersis.cdc11}. Here the use of binary information allows us to conclude asymptotic convergence without the need to dynamically update the quantizer resolution. The use of binary information in coordination problems (\cite{cortes.aut06,chen.et.al.aut11}) has been proven useful to the design and real-time implementation of distributed controls for systems of first- or second-order agents in a cyber-physical environment (see e.g.~\cite{GS-DVD-KHJ:13,CN-JC:12,depersis.frasca.necsys12}). We envision that a similar role will be played by the results in this paper for a  larger class of coordination problems (see \cite{CDP:PF:JH:CDC13} for an early result in this respect).\\
  Another advantage of our approach is that the resulting control laws are implemented by very simple directional commands (such as ``move north", ``move north-east", etc. or ``stay still"). We also show that the presence of coarse information does not affect the ability of the proposed controllers to achieve the formation in a leader-follower setting in which the prescribed reference velocity is only known to the leader. This paper adopts a similar setting as in \cite{depersis.cdc11,depersis.jaya} but controllers and analysis are different. Moreover, the paper investigates the formation control problem with unknown reference velocity tracking and matched disturbance rejection that was not considered in \cite{depersis.cdc11,depersis.jaya}. Compared with \cite{Yu.et.al.2012}, where also coarse information was used for rendezvous, the results in our contribution apply to a different class of systems and to a different cooperative control problem. Early results with the same sensing scenario but for a formation of agents modeled as double integrators have been presented in \cite{mj.cdp.acc13}.\\ 
The paper is organized as follows: Section~\ref{sec:pp} introduces the problem statement along with some motivations and notation. Analysis of the formation keeping problem with coarse data in the case of known/unknown reference velocity is studied in Section ~\ref{sec:vel}. Section \ref{sec:disturbance} investigates the problem of formation keeping with coarse data in the presence of matched disturbances. Related simulations are presented in Section~\ref{sec:sim}. The paper is summarized in Section~\ref{sec:conclusion}.

{\it Notation.} Given two sets $S_1, S_2$, the symbol $S_1\times S_2$ denotes the Cartesian product of two sets. This can be iterated. The symbol 
$\bigcart_{k=1}^m S_k$ 
denotes $S_1 \times S_2 \times \ldots \times S_m$. For a set $S$, ${\rm card}(S)$ denotes the cardinality of the set $S$. Given a matrix $M$ of real numbers, we denote by ${\mathcal R} (M)$ and ${\mathcal N} (M)$ the range and the null space, respectively. The symbols $\mathbf{1}, \mathbf{0}$ denotes vectors or matrices of all $1$ and $0$ respectively. Sometimes the size of the matrix is explicitly given. Thus, $\mathbf{1}_n$ is the $n$-dimensional vector of all $1$. $I_{p}$ is the $p \times p$ identity matrix. Given two matrices $A,B$, the symbol $A\otimes B$ denotes the Kronecker product. 

\section{Preliminaries}
\label{sec:pp}
\subsection{The multi-agent system} 
In this subsection we review the passivity-based approach to multi-agent system control (\cite{arcak.tac07}, see also \cite{buerger.et.al.cdc11,muenz.et.al.tac11,depersis.jaya}). A network of $N$  agents in $\R^{p}$ are considered. For each agent $i$, $x_{i} \in \R^{p}$ represents its position. The communication topology of the network is assumed to be modeled by a connected and undirected graph $G = (V, E)$, where $V$ is a set of $N$ nodes and $E \subseteq V \times V$ is a set of $M$ edges connecting the nodes. 
Label one end of each edge in $E$ with a positive sign and the other end with a negative sign. We define the relative position $z_{k}$ as following
\begin{equation*}
z_{k} = \left\{
\begin{array}{c}
 x_{i} - x_{j}\;\;\;\;\text{if node $i$ is the positive end of the edge k}\\
 x_{j} - x_{i}\;\;\;\;\text{if node $j$ is the positive end of the edge k}\\
\end{array}
\right.
\end{equation*}
where $x_i\in \R^p$ is the position of the agent $i$ expressed in an inertial frame. We define the $N \times M$ incidence matrix $B$ associated with the graph $G$ as follows
\begin{equation*}
b_{ik} = \left\{
 \begin{array}{c}
 +1\;\;\;\;\text{If node $i$ is the positive end of the edge k}\\
 -1\;\;\;\;\text{If node $j$ is the positive end of the edge k}\\
 0\;\;\;\;\;\text{otherwise.}\hspace{45mm}\\
 \end{array}
\right.
\end{equation*}
By definition of $B$, we can represent the relative position variable $z$, with $z \triangleq [z_{1}^{T} \ldots z_{M}^{T}]^{T}$, $z \in \R^{p M}$,  
as a function of the position variable $x$, namely
\begin{equation}
z = (B^{T} \otimes I_{p}) x.
\label{eq:a8}
\end{equation}
Equation (\ref{eq:a8}), implies that $z$ belongs to the range space $\mathcal{R} (B^{T} \otimes I_{p})$.\\
Each agent is expected to track its desired (time-varying) reference velocity denoted by $v_{i}^{r}$. Define the velocity error 
\begin{equation}
y_{i} = \dot{x_{i}} - v_{i}^{r}.
\label{eq:a4}
\end{equation}
 The dynamics of each agent is given by 
\begin{equation}
\mathcal{H}_{i}: \left\{
 \begin{array}{c}
  \dot{\xi_{i}} = f_{i} ( \xi_{i} ) + g_{i} ( \xi_{i} ) u_{i}\\
  y_{i} = h_{i} ( \xi_{i} )\hspace{12mm}\\
 \end{array}
\right.
\label{eq:a5}
\end{equation}
where $\xi_{i} \in \R^{n_{i}}$ is the state variable, $u_{i} \in \R^{p}$ is the control input, $y_{i} \in \R^{p}$ is the velocity error, and the exogenous signal $v_{i}^{r} \in \R^{p}$ is the reference velocity for the agent $i$. The maps $f_{i}$, $g_{i}$ and $h_{i}$ are assumed to be locally Lipschitz such that $f_{i} ( \mathbf 0 ) = \mathbf 0$, $g_{i} ( \mathbf 0 )$ is full column-rank, and $h_{i} ( \mathbf 0 ) = \mathbf 0$.\\ 
The system $\mathcal{H}_{i}$ is assumed to be strictly passive from the input $u_{i}$ to the velocity error $y_{i}$. Since the system $\mathcal{H}_{i}$ is strictly passive, there is a continuously differentiable storage function $S_{i}: \R^{n_{i}} \rightarrow \R_{+}$ which is positive definite and radially unbounded and satisfies
\begin{equation}
\frac{\partial S_{i}}{\partial \xi_{i}} [ f_{i}( \xi_{i} ) + g_{i} ( \xi_{i} ) u_{i} ] \leq - W_{i} ( \xi_{i} ) + {y_{i}}^{T} u_{i}
\label{eq:passivity}
\end{equation}
where $W_{i}$ is a continuous positive function and $W_{i} ( \mathbf 0 ) = \mathbf 0$.\\
For the sake of conciseness, the equations (\ref{eq:a4}) and (\ref{eq:a5}) are written in the  compact form
\begin{equation}
\begin{split}
 \dot{x} &= \underbrace{ \left(
 \begin{array}{c}
 h_{1} ( \xi_{1} )\\
 \hspace{0.5mm}\vdots\\
 h_{N} ( \xi_{N} )
 \end{array}\right)}_{h ( \xi )} + \underbrace{\left(
 \begin{array}{c}
 v_{1}^{r}\\
 \hspace{0.5mm}\vdots\\
 v_{N}^{r}
 \end{array} \right)}_{v^r}
 \\
 \dot{\xi} &= \underbrace{ \left(
 \begin{array}{c}
 f_{1} ( \xi_{1} )\\
 \hspace{0.5mm}\vdots\\
 f_{N} ( \xi_{N} )
 \end{array} \right)}_{f ( \xi )} + \underbrace{ \left(
 \begin{array}{ccc}
 g_{1} ( \xi_{1} ) & \ldots & \mathbf 0 \\
 \vdots & \ddots & \vdots\\
 \mathbf 0 & \ldots & g_{N} ( \xi_{N} ) \end{array} \right)}_{g ( \xi )} u
\\
y &= h(\xi).
\end{split}
\label{eq:a7}
\end{equation}
We define the concatenated vectors $x \triangleq [x_{1}^{T} \ldots x_{N}^{T}]^{T}$, $x \in \R^{p N}$, $x_{i} \in \R^{p}$, and $\xi \triangleq [\xi_{1}^{T} \ldots \xi_{N}^{T}]^{T}$, $\xi \in \R^{p N}$, $\xi_{i} \in \R^{p}$. The desired position  $x^{*}$, with $x^{*} \triangleq [{x_{1}^{*}}^{T} \ldots {x_{N}^{*}}^{T}]^{T}$, $x^{*} \in \R^{p N}$, $x_{i}^{*} \in \R^{p}$ is given by $x^* = x^*_0 +\mathbf{1}_N\otimes \int_0^t v^*(s) ds$, with $x^*_0$ a constant vector. Analogously, $z^{*} \triangleq [{z_{1}^{*}}^{T} \ldots {z_{M}^{*}}^{T}]^{T}$, $z^{*} \in \R^{p M}$,  $z_{i}^{*} \in \R^{p}$, is the desired relative position vector. The two vectors are related by the identity $z^{*}=(B^T\otimes I_p) x^*$.
%
\subsection{Problem Statement} 
Our goal is to design coordination control laws to attain the following collective behaviors for the formation of agents (\ref{eq:a4}), (\ref{eq:a5}): 
\begin{enumerate}[(i)]
\item The velocity of each agent of the network asymptotically converges to the desired  reference velocity ${v^*}(t) \in \R^{p}$
\begin{equation}
\lim_{t\to\infty} || \dot{x_{i}}(t) - {v^*}(t) || = \mathbf 0,\;\; i = 1,...,N.
\label{eq:a2}\end{equation}
\item Each relative position vector $z_{k}$ converges to the desired relative position vector $z_{k}^{*}$
\begin{equation}
\lim_{t\to\infty} || z_k - z_{k}^{*}||= \mathbf 0, \;\; z_{k}^{*} \in \R^{p},\; k = 1,...,M.
\label{eq:a3}
\end{equation}
\item 
In the presence of matched disturbances, i.e.~in the case in which the equation (\ref{eq:a5}) is replaced by
\[
\dot{\xi_{i}} = f_{i} ( \xi_{i} ) + g_{i} ( \xi_{i} ) (u_{i}+d_i)
\]
where $d_i$ is a disturbance signal generated by a suitable  exosystem (see Section \ref{sec:disturbance}), the collective behaviors in (i) and (ii) are still guaranteed. 
\end{enumerate}
Differently from other work in formation control,  we are interested in control laws that achieve complex coordination tasks using very coarse information about the relative positions of the agents. We assume that the agents of the network are equipped with sensors that are capable to detect whether the relative position of two agents is above or below a prescribed one.  
In line with our  goals, for each agent $i$, the control law is designed as
\begin{equation}
u_{i} = -\sum_{k = 1}^{M} b_{ik}\ \sign (z_{k} - z_{k}^{*}),\;\;\;i = 1 \;...\; N,
\label{eq:b1}
\end{equation}
where $\sign(\cdot)$ is the sign function and operates on each element of the $p$-dimensional vector $z_{k} - z_{k}^{*}$. The sign function $\sign(\cdot): \R \rightarrow \{ -1, +1 \}$ is defined as follows: 
\[
\sign( \zeta )=\left\{\ba{lll}
+1 & {\rm if} & \zeta \geq 0\\
-1 & {\rm if} & \zeta < 0.
\ea\right.
\]
Observe that, $b_{ik} \neq 0$ if and only if the edge $k$ connects the agent $i$ to one of its neighbors. This implies the control law $u_{i}$ uses only the information available to the agent $i$. The above control law is inspired by the passivity-based control design proposed in \cite{arcak.tac07}. 
%

\noindent {\it Discussion about the control (\ref{eq:b1}).}
The proposed control (\ref{eq:b1}) has the following  interpretation. Let the position of the agents be given with respect to an inertial frame in $\R^{p}$. Assign to each agent a local Cartesian coordinate system with an orientation identical to the inertial frame such that the current position of the agent is the origin of its local coordinate system. Therefore, there are $p$ 
mutually perpendicular hyperplanes intersecting at the origin of the agent's local coordinate system and  partitioning the state space into $2^{p}$ hyperoctants. 
For each neighbor of the agent $i$, the corresponding $p$-dimensional vector $z_{k} - z_{k}^{*}$ lies in one of these hyperoctants (possibly at the boundary) and contributes the term $- b_{ik} \sign (z_{ik} - z_{ik}^{*})$ to the control $u_i$. 
The implementation of (\ref{eq:b1}) only requires the agent to detect when $z_{k} - z_{k}^{*}$ is crossing the boundary between hyperoctants.\footnote{In that regard, the control law is reminiscent of the controllers of \cite{Yu.et.al.2012} that changes when a neighbor  enters or leaves 	the field-of-view sector of the agent.} 
A few other advantages of this sensing scenario are discussed below: \\
(i) The formation is achieved with large inaccuracies in the measurements of $z_k-z_{k}^\ast$. As a matter of fact, the calculated control vector corresponding to a given vector $-b_{ik}(z_k-z_{k}^\ast)$ 
is the same no matter where $(z_k-z_{k}^\ast)$ precisely lies in the hyperoctant.\\
(ii) If the agent $i$ detects that the vector $z_{k} - z_{k}^{*}$ 
is crossing the boundary of any two hyperoctants, then a new control vector is applied  due to the term $- b_{ik}\sign (z_{k} - z_{k}^{*})$ in (\ref{eq:b1}).
In our analysis below, however, we show that the exact formation is still achieved if this control vector at the boundary takes any value in the convex hull of the two control vectors associated with the two hyperoctants. In this respect the control (\ref{eq:b1}) is robust to possible uncertainties in the precise detection of the boundary crossing.  \\
(iii) The proposed control law guarantees the achievement of a desired formation by very simple directional commands, namely move closer if the actual position is larger than the desired one, move away if smaller.\\
(iv) The input (\ref{eq:b1}) provides a {\em finite}-valued control that uses {\em binary} measurements. In the case of networks of kinematic agents, the finite-valued  control law of \cite{cortes.aut06} 
inspired the design of self-triggered control algorithms (see \cite{depersis.frasca.necsys12}). These algorithms are shown to reduce the amount of time during which sensors are active and to reduce the amount of information exchanged among the agents. Moreover, they can be used to overcome possible fast switching of the controllers (see discussion in Section \ref{sec:sim}).  
The control investigated in this paper can pave the way towards the design of self-triggered cooperative control of nonlinear systems.  A first step in this regard has been taken in \cite{CDP:PF:JH:CDC13}. 
%

\section{Analysis}
\label{sec:vel}
In this section we investigate the stability properties of the closed-loop system. We define the variable $\tilde{z} = z - z^{*}$ and write the control laws $u_{i}$, $i=1,2,\ldots, N$, in the following compact form
\begin{equation}
u = - (B \otimes I_{p}) \sign\ \tilde{z}.
\label{eq:u}
\end{equation}
 From equation (\ref{eq:a8}), we can write $\tilde{z} = ( B^{T} \otimes I_{p} ) ( x - x^{*} )$.
 We consider the stability of the origin of the error system, with 
 state 
 variables $[ {\xi}^{T}\ \ {\tilde{z}}^{T} ]^{T}$. To derive the error system equations, we 
 take the derivative of $\tilde z$, thus obtaining 
\begin{equation}
\dot{\tilde{z}} = ( B^{T} \otimes I_{p} ) ( \dot{x} - {\dot x}^{*} ).
\label{eq:b3}
\end{equation}
From equation (\ref{eq:a7}), we further have
\begin{equation}
\dot{x} - {\dot x}^{*} = h(\xi) + v^r - {\mathbf 1}_{N} \otimes v^*
\label{eq:b4}
\end{equation}
where $v^*$ is the desired reference velocity for the  formation. By a property of the incidence matrix of a connected graph, $( B^{T} \otimes I_{p} ) ( {\mathbf 1}_{N} \otimes v^* ) = \mathbf{0}$. 
Thus, we represent the dynamics of the error system in  the following general form
\begin{equation}
\begin{array}{c}
\begin{aligned}
 \dot{\tilde{z}} &= (B^{T} \otimes I_{p}) {\dot{x}}= (B^{T} \otimes I_{p}) (h(\xi) + v^r)\\[3mm]
 \dot{\xi} &= f(\xi) + g(\xi) (-B \otimes I_{p}) \sign\ \tilde{z}. 
\end{aligned}
\end{array}
\label{eq:b5}
\end{equation}
The system (\ref{eq:b5}) has a discontinuous right-hand side due to the discontinuity of the sign function at zero. Before analyzing the system, we first define an appropriate notion of the solution. In this paper,  the solutions to the system above are intended in the Krasowskii sense. 
As in \cite{ceragioli.et.al.aut11,depersis.jaya}, the motivation to consider these solutions lies in the fact that there exists convenient Lyapunov methods to analyze their asymptotic behavior and that they include other notions of solutions such as Carath\'eodory solutions if the latter exist. 
Let $X = ( \tilde z, \xi)$ and let $F(X)$ be the set-valued map 
\begin{equation}
\ba{rcl}
F (X) &=& \left(\begin {array}{c}
 (\!B^T\otimes I_p\!) (h(\xi) + v^r)\\
 f(\xi)
 \end{array}\right) 
-\left(\begin{array}{c}
 \mathbf{0}_{Np\times 1}\\
 g(\xi) (\!B\otimes I_p\!) 
 \end{array}\right){\mathcal K} \sign\ \tilde z
\label{eq:F}
\ea
\end{equation}
where ${\mathcal K}\; \sign\ \tilde z = \bigcart_{k=1}^{ M\stopmodif} \bigcart_{\ell=1}^{ p\stopmodif}{\mathcal K}\; \sign\ \tilde z_{k\ell}$ and 
\begin{equation}
{\mathcal K}\; \sign\ \tilde z_{k\ell} = \left\{\begin{array}{lll}
\{\sign\ \tilde z_{k\ell}\} & if & \tilde z_{k\ell}\ne 0\\
~[-1,1] & if & \tilde z_{k\ell}= 0.
\end{array}\right.
\end{equation}
We define $X(t)=(\tilde z(t), \xi(t))$ a Krasowskii solution to (\ref{eq:b5}) on the interval $[0,\;t_1]$ if it is an absolutely continuous function  which satisfies the differential inclusion 
\begin{equation}
\dot {X}(t) \in F(X( t ))
\label{eq:b7}
\end{equation}
for almost every $t\in [0, t_1]$, with $F$ defined as in (\ref{eq:F}), and
\begin{equation}
{\mathcal K}(F(X)) = \bigcap_{\delta >0} \overline{{\rm co}} (F(B(x,\delta))),
\label{eq:KF}
\end{equation}
with $\overline{{\rm co}}$ the closed convex hull of a set and $B(x,\delta)$ the ball centered at $x$ and of radius $\delta$. Local existence of Krasowskii solutions to the differential inclusion above is always guaranteed (see e.g.~\cite{hajek} and \cite{ceragioli.et.al.aut11}, Lemma 3). 
For the convenience of the reader, a brief reminder about non-smooth control theory tools adopted in this paper is given in Appendix A. 
\subsection{Known reference velocity}
In this section, we consider the case where the desired reference velocity, ${v^*}(t)$ is known to all the agents, namely $ v^r = {\mathbf 1}_{N} \otimes v^*$. Then, the closed-loop system simplifies as
\begin{equation}
\begin{array}{c}
\begin{aligned}
 \dot{\tilde{z}} &= ( B^{T} \otimes I_{p} ) h(\xi)\hspace{25mm}\\[3mm]
 \dot{\xi} &= f( \xi ) + g( \xi ) ( - B \otimes I_{p}) \sign\ \tilde{z}.
\end{aligned}
\end{array}
\label{eq:cl1}
\end{equation}
\begin{proposition}\label{p1}
Any Krasowskii solution to (\ref{eq:cl1}) exists for all $t \geq 0$ and converges to the origin.
\end{proposition}
\begin{proof}
The proof is based on the application of the non-smooth La Salle's invariance principle (\cite{bacciotti.ceragioli.esaim99}). Consider the following locally Lipschitz Lyapunov function
\[
V(\tilde z, \xi)= ||\tilde z||_1+ S(\xi)
\]
where $||\tilde z||_1$ denotes the $1$-norm $||\tilde z||_1=\sum_{k,\ell} |\tilde z_{k\ell}|$ and evaluate its set-valued derivative $\dot{\overline V}(\tilde z, \xi)$ along (\ref{eq:cl1}). We have
\begin{equation*}
\dot{\overline V}(\tilde z, \xi)= \{a\in \R\,:\, \exists w\in F(\tilde z, \xi)\ \ {\rm s.t.}\ 
a=\langle w, p \rangle,\  \mbox{for all} \ \ p\in \partial V (\tilde z, \xi)\}.
\end{equation*}
By the definition of $F(\tilde z, \xi)$ in (\ref{eq:F}), for any 
$w\in F(\tilde z, \xi)$ there exists $w^{\tilde z}\in {\mathcal K}\sign\;\tilde z$ such that
\[
w= \left(\ba{c}
(B^T\otimes I_p) h(\xi)\\
f(\xi)
\ea\right)
-\left(\ba{c}
\mathbf{0}\\
g(\xi)\ (B\otimes I_p) 
\ea\right)w^{\tilde z}.
\]
Observe that the Clarke generalized gradient $\partial V (\tilde z, \xi)$ is given by
\begin{equation*}
\partial V (\tilde z, \xi)=\{
p: p=\left(\ba{c} p^{\tilde z}\\ \nabla S(\xi)\ea\right)\; {\rm s.t.} \ \ 
p^{\tilde z}_{k\ell}\in \left\{\ba{lll}\{\sign\;{\tilde z}_{k\ell}\} & if & {\tilde z}_{k\ell}\ne 0\\
~[-1,+1] & if & {\tilde z}_{k\ell}= 0
\ea\right.\}.
\end{equation*}
Suppose that $\dot{\overline V}(\tilde z, \xi)\ne\emptyset$ and take $a\in \dot{\overline V}(\tilde z, \xi)$. Then by definition there exists $w\in F(\tilde z, \xi)$ such that $a=\langle w, p \rangle$ for all $p\in \partial V (\tilde z, \xi)$. 
Choose $p\in \partial V (\tilde z, \xi)$ such that $p^{\tilde z}=w^{\tilde z}$. Thus
\[\ba{rcl}
a&=&\langle \left(\!\!\ba{c}
(B^T\otimes I_p) h(\xi)\\
f(\xi)
\ea\!\!\right)
-\left(\!\!\ba{c}
\mathbf{0}\\
g(\xi)\ (B\otimes I_p) 
\ea\!\!\right)w^{\tilde z},
\left(\!\!\ba{c} w^{\tilde z}\\ \nabla S(\xi)\ea\!\!\right)
\rangle\\\\
&=&
\langle (\!B^T\otimes I_p\!) h(\xi), w^{\tilde z}\rangle
+\langle f(\xi) - g(\xi) (\!B\otimes I_p\!) w^{\tilde z}, \nabla S(\xi)\rangle\\\\
&\leq& -W(\xi) + 
\langle h(\xi), (B\otimes I_p) w^{\tilde z}\rangle
 + \langle (B^T\otimes I_p) h(\xi), w^{\tilde z}\rangle
 \\[2mm]
  &\leq& -W(\xi)
\ea\]
Hence, for any state $(\tilde z, \xi)$ such that $\dot{\overline V}(\tilde z, \xi)\ne\emptyset$, $\dot{\overline V}(\tilde z, \xi)=\{a \in \R: a \leq -W(\xi)\}\subseteq (-\infty, 0]$. 
Therefore, the solutions can be extended for all $t\ge 0$ and by La Salle's invariance principle they converge to the largest weakly invariant set where $\xi=\mathbf{0}$. Since $ v^r = {\mathbf 1}_{N} \otimes v^*$, from (\ref{eq:F}), any point $(\tilde z, \mathbf{0})$ on this invariant set must necessarily satisfy
\[
\mathbf{0} =\left(\ba{c}
\mathbf{0}\\
g(\mathbf{0}) (B\otimes I_p) 
\ea\right)
w^{\tilde z}
\]
for some $w^{\tilde z}\in {\mathcal K} \sign\;\tilde z$. Since $g(\mathbf{0})$ is full-column rank, this implies 
$w^{\tilde z} \in {\mathcal N}(B\otimes I_p)$. Bearing in mind that $\tilde z\in {\mathcal R}(B^T\otimes I_p)$, then $\langle \tilde z, w^{\tilde z}\rangle=0$. 
We claim that necessarily $\tilde z=\mathbf{0}$. In fact, by contradiction, $\tilde z\ne \mathbf{0}$ would imply the existence of at least a component $\tilde z_{k\ell}\ne 0$. As $\tilde z_{k\ell}\ne 0$ necessarily implies $w^{\tilde z}_{k\ell}\ne 0$ by definition of $w^{\tilde z}\in {\mathcal K}\sign\; \tilde z$, then it would imply $\tilde z_{k\ell} w^{\tilde z}_{k\ell}> 0$, which shows the contradiction. Hence, on the invariant set, $\tilde z=\mathbf{0}$ and all the system's solutions converge to the origin.
\end{proof}
\subsection{Unknown reference velocity}\label{sec.urv}
In the previous section, we assumed that the desired reference velocity is known to all of the agents. This assumption is very restrictive and one wonders whether it is possible to relax this assumption despite the very coarse information scenario we are confining ourselves. The positive answer to this question is given in the analysis below.\\
The problem of the reference velocity unknown to the agents is tackled here in the scenario in which at least one of the agents is aware of $v^\ast$. We refer to this agent as the leader (\cite{ren.book}, \cite{mesbahi.egerstedt.book}). As in \cite{bai.et.al.book}, we further assume that the class of reference velocities that the formation can achieve are those generated by an autonomous system that we refer to as the exosystem. Given two matrices $\Phi, \Gamma^v$, whose properties will be made precise later on, the exosystem obeys the following equations
\be\label{exosv}
\dot {w^v}= \Phi {w^v},\quad {v^\ast}=\Gamma^v {w^v}.
\ee 
Taking inspiration from the theory of output regulation (see e.g.~\cite{isidori.et.al.book}), 
an internal-model-based controller can be adopted for each agent $i=2,3,\ldots, N$, namely 
\be\label{imv} 
\ba{rcll}
\dot{\eta}_i &=& \Phi \eta_i +G \tilde u_i\\
v^r_i &=& \Gamma^v \eta_i, & i= 2,3,\ldots, N.
\ea\ee
When $\breve u_i=\mathbf 0$ and the system is appropriately initialized, the latter system is able to generate any $w^v$ solution to (\ref{exosv}). Because $v^\ast$ is known to the agent $1$, there is no need to implement system (\ref{imv}) at the agent $1$. The input $\tilde u_i$, as well as the input matrix $G$ are to be designed later. Define ${v^r}=({v^\ast}^T\ {v^r_2}^T\ \ldots\ {v^r_n}^T)$ and the new error variables\\
\[
\tilde \eta_i =\eta_i -w^v,\quad \tilde v_i = v^r_i - v^\ast,\quad i=2,\ldots, N,
\]
along with the fictitious variables $\tilde \eta_1 =\mathbf 0$, $\tilde v_1=\mathbf 0$. We obtain
\[\ba{rcl}
\dot{\tilde \eta} =
\left(\ba{c}
\dot{\tilde \eta}_1\\
\dot{\tilde \eta}_2\\
\ldots\\
\dot{\tilde \eta}_n
\ea\right)
=
\underbrace{
\left(\ba{cccc}
\mathbf{0} & \mathbf{0} & \ldots & \mathbf{0}\\
\mathbf{0} & \Phi & \ldots & \mathbf{0}\\
\vdots & \vdots & \ddots & \vdots \\
\mathbf{0} & \mathbf{0} & \ldots & \Phi\\
\ea\right)}_{\tilde \Phi}
\tilde \eta + 
\underbrace{
\left(\ba{cccc}
\mathbf{0} & \mathbf{0} & \ldots & \mathbf{0}\\
\mathbf{0} & G & \ldots & \mathbf{0}\\
\vdots & \vdots & \ddots & \vdots \\
\mathbf{0} & \mathbf{0} & \ldots & G\\
\ea\right)}_{\tilde G}
\underbrace{
\left(\ba{c}
\tilde u_1\\
\tilde u_2\\
\ldots\\
\tilde u_n
\ea\right)}_{\tilde u} \\\\
\tilde v =
\underbrace{\left(\ba{cccc}
\mathbf{0} & \mathbf{0} & \ldots & \mathbf{0}\\
\mathbf{0} & \Gamma^v & \ldots & \mathbf{0}\\
\vdots & \vdots & \ddots & \vdots \\
\mathbf{0} & \mathbf{0} & \ldots & \Gamma^v\\
\ea\right)}_{\tilde \Gamma^v}\tilde \eta \hspace{65mm}
\ea\]
and in compact form
\be\label{im.tilde} 
\ba{rcll}
\dot{\tilde \eta} &=& {\tilde{\Phi}} {\tilde {\eta}} + \tilde G \tilde u\\
\tilde v &=& {\tilde{\Gamma^v}} {\tilde {\eta}}.
\ea\ee
From (\ref{eq:b5}), the error position vector can be written as 
\begin{equation*}
\ba{rcl}
\dot{\tilde z} &=& (B^T\otimes I_p) (h(\xi)+v^r) \\
&=& (B^T\otimes I_p) (h(\xi)+v^r-\mathbf{1}_n\otimes v^\ast+\mathbf{1}_n\otimes v^\ast)\\
&=&(B^T\otimes I_p) (h(\xi)+\tilde v).\hspace{32mm}
\ea
\end{equation*}
The overall closed-loop system is
\be\label{ocls}\ba{rcl}
\dot {\tilde z} &=& (B^T\otimes I_p) (h(\xi) +\tilde \Gamma^v \tilde \eta)\\
\dot{\xi} &=& f( \xi ) + g( \xi ) u\\
\dot{\tilde \eta} &=& \tilde \Phi \tilde \eta +\tilde G \tilde u.\\
\ea\ee
\begin{proposition}\label{p2}
Assume that $\Phi$ is skew-symmetric, namely
\be\label{Phi}
\Phi^T+\Phi=\mathbf{0},
\ee
$(\Gamma^v, \Phi)$ is an observable pair 
and let
\be\label{breve.G}
 G={\Gamma^v}^T.
\ee
Then all the Krasowskii solutions to (\ref{ocls}) in closed-loop with the control input 
\be\label{controlwd}
u=\tilde u = - (B\otimes I_p) \sign\; \tilde z
\ee
converge to the point $\tilde z=\mathbf{0}$, $\xi=\mathbf{0}$, $\tilde \eta=\mathbf{0}$. 
\end{proposition}

\begin{proof}
Any Krasowskii solution to (\ref{ocls}) with the control law defined by (\ref{controlwd}) satisfies a differential inclusion of the form (\ref{eq:F}) where $X= (\tilde z,\xi, \tilde \eta)$ and 
\begin{equation*}
F(\tilde z,\xi, \tilde \eta)=\left(\ba{c}
(B^T\otimes I_p) (h(\xi)+\tilde \Gamma ^v \tilde \eta) \\
f(\xi) \\
\tilde \Phi\tilde \eta
\ea\right) - 
\left(\ba{c}
\mathbf{0}\\
g(\xi) (B\otimes I_p) \\
\tilde G (B\otimes I_p) 
\ea\right){\mathcal K}\,\sign\; \tilde z.
\end{equation*}
Let
\[
V(\tilde z,\tilde v, \tilde \eta) = 
||\tilde z||_1+ S(\xi)+ \frac{1}{2} \tilde \eta^T \tilde \eta.
\]
For any
$w\in F(\tilde z, \xi, \tilde \eta)$, there exists $w^{\tilde z}\in {\mathcal K}\sign\; \tilde z$ such that
\[
w=
\left(\ba{c}
(B^T\otimes I_p) (h(\xi)+\tilde \Gamma ^v \tilde \eta) \\
f(\xi)\\
\tilde \Phi\tilde \eta
\ea\right)
-\left(\ba{c}
\mathbf{0}\\
g(\xi) (B\otimes I_p) \\
\tilde G\;(B\otimes I_p) 
\ea\right)w^{\tilde z}.
\]
Moreover for any $p\in \partial V(\tilde z, \xi, \tilde \eta)$ there exists $p^{\tilde z}\in \partial ||\tilde z||_1$, with 
\[
p^{\tilde z}_{k\ell}\in 
\left\{\ba{lll}
\{\sign\;{\tilde z}_{k\ell}\} & if & {\tilde z}_{k\ell}\ne 0\\
~[-1,+1] & if & {\tilde z}_{k\ell}= 0,
\ea\right.
\]
and $k=1,2,\ldots,m$, $\ell=1,2,\ldots, p$, such that
\[
p=\left(\ba{c}
p^{\tilde z}\\ \nabla S(\xi) \\ \tilde \eta
\ea\right).
\]
For $w^{\tilde z}\in {\mathcal K}\sign\; \tilde z$, let $p\in \partial V(\tilde z, \xi, \tilde \eta)$ such that $w^{\tilde z}=p^{\tilde z}$. Hence
\[
\ba{l}
\langle w, p\rangle = \langle (B^T\otimes I_p) (h(\xi)+\tilde \Gamma ^v \tilde \eta), w^{\tilde z}\rangle \\[2mm]
\hspace{1.5cm} +\langle f(\xi) - g(\xi) (B\otimes I_p)w^{\tilde z}, \nabla S(\xi)\rangle\\[2mm]
\hspace{1.5cm} +\langle \tilde \Phi\tilde \eta-\tilde G\; (B\otimes I_p)w^{\tilde z}, \tilde \eta\rangle. \\[2mm]
\ea
\]
By the conditions (\ref{Phi}) and (\ref{breve.G}), the equality above simplifies as 
\[
\ba{l}
\langle w, p\rangle =\langle(B^T\otimes I_p) h(\xi), w^{\tilde z}\rangle + \langle f(\xi) - g(\xi) (B\otimes I_p)w^{\tilde z}, \nabla S(\xi)\rangle.
\ea
\]
Since the dynamic of each agent is strictly passive, we have
\begin{equation*}
\langle f(\!\xi) - g(\!\xi) (\!B\otimes I_p\!)w^{\tilde z}, \nabla S(\!\xi)\rangle \leq -W(\!\xi)
-\langle h(\!\xi), (\!B\otimes I_p\!)w^{\tilde z}\rangle
\end{equation*}
which implies $\langle w, p\rangle \leq -W(\xi)$.\\ 
Hence, similarly to the proof of Proposition \ref{p1}, $\dot{\overline{V}}(\tilde z, \xi, \tilde \eta)=\{a \in \R: a \leq -W(\xi)\} \subseteq (-\infty, 0]$, provided that $\dot{\overline{V}}(\tilde z, \xi, \tilde \eta)\ne \emptyset$. Because $V$ is positive definite and  proper, solutions are bounded and one can apply La Salle's invariance principle, which shows that every Krasowskii solution converges to the largest weakly invariant set where $W(\xi)=\mathbf{0}$. By definition of $W(\xi)$, the latter is equivalent to have  convergence to the largest weakly invariant set where $\xi=\mathbf{0}$. To characterize this set, we look for points $(\tilde z, \mathbf{0}, \tilde \eta)$ such that $(\dot{\tilde z}, \mathbf{0}, \dot {\tilde \eta})\in F(\tilde z, \mathbf{0}, \tilde \eta)$. Having $f(\mathbf{0})=h(\mathbf{0})=\mathbf{0}$, this means that we seek points $(\tilde z, \mathbf{0}, \tilde \eta)$ for which there exists $w^{\tilde z}\in{\mathcal K}\,\sign\; \tilde z$ such that 
\[
\left(\!\ba{c}
\dot{\tilde z}\\
\mathbf{0}\\
\dot{\tilde\eta}
\ea
\!\right)
=
\left(\!\ba{c}
(B^T\otimes I_p) (\tilde \Gamma ^v \tilde \eta) \\
\mathbf{0}\\
\tilde \Phi\tilde \eta
\ea\!\right)
-\left(\!\ba{c}
\mathbf{0}\\
g(\mathbf{0}) (B\otimes I_p) \\
(\tilde \Gamma ^v)^T (B\otimes I_p) 
\ea\!\right)w^{\tilde z}. 
\]
The equality  $\mathbf{0}= -g(\mathbf{0})(B\otimes I_p)w^{\tilde z}$, implies that necessarily $\tilde z=\mathbf{0}$ (see the final part of the proof of Proposition \ref{p1}).
This also implies that on the invariant set the system evolves as  
\[\ba{rcl}
\mathbf{0} &=& (B^T\otimes I_p) (\tilde \Gamma ^v \tilde \eta)\\
\dot{\tilde \eta} &=& \tilde \Phi \tilde \eta.
\ea\]
From this point on, the argument is the same as in \cite{bai.et.al.book}. Indeed, 
as ${\mathcal N}(B^T)={\mathcal R}(\mathbf{1}_n)$, by the block diagonal structure of $\tilde \Gamma ^v $, we obtain that all the sub vectors of $\tilde \Gamma ^v \tilde \eta$ must be the same and equal to zero, since the first $p$ component of $\tilde \Gamma ^v \tilde \eta$ are identically zero. In other words, $\Gamma^v \tilde \eta_i=\mathbf{0}$ for all $i=2,3, \ldots, N$. Hence we conclude that on the largest weakly invariant set where $\xi=\mathbf{0}$, we have 
\[\ba{rcll}
\mathbf{0} &=& \Gamma ^v \tilde \eta_i\\
\dot{\tilde \eta}_i &=& \Phi \tilde \eta_i, & i=2,3, \ldots, N.
\ea\]
By the observability of $(\Gamma^v, \Phi)$, it holds that $\tilde \eta_i=\mathbf{0}$ for all $i=2,3, \ldots, N$. This finally proves the claim. 
\end{proof}

\section{Formation control with matched disturbance rejection}
\label{sec:disturbance}
In this section, we study the problem of formation keeping with very coarse exchanged information in the presence of matched input disturbances. 
We consider a formation  of strictly passive agents, where the dynamics of each agent of the networked system is
\begin{equation*}
\dot{x}_{i}= y_i + v_{i}^{r}(t)\hspace{9mm}
\end{equation*}
\be\mathcal{H}_{i}: \left\{
 \begin{array}{c}
 \begin{aligned}
  \dot{\xi_{i}} &= f_{i} (\xi_{i}) + g_{i} (\xi_{i}) (u_i + d_i)\\
  y_{i} &= h_{i} (\xi_{i}),\\
  \end{aligned}
 \end{array}
\right.
\label{eq:d1}
\ee
where $d_i$ is a disturbance. 
The network should converge to the prescribed formation and evolve with the given desired reference velocity $v^\ast$ despite the action of the disturbance $d$. As in the previous section, we first consider the case in which  the reference velocity is unknown to all the agents except one. As explained in Section ~\ref{sec.urv}, we adopt an internal-model-based controller to recover the desired reference velocity. Similarly we suppose that the disturbance signal $d_i$ at the agent $i$ is generated by an exosystem of the form 
\be\label{exos.d}
\dot w^d_i= \Phi^d_i w^d_i,\quad {d_i}=\Gamma^d_i w^d_i, \quad i=1,2,\ldots,N.
\ee 
To counteract the effect of the disturbance, we introduce an additional internal-model-based controller given by 
\be\label{im.d} 
\ba{rcll}
\dot{\theta}_i &=& \Phi^d_i \theta_i +G^d_i \check{u}_i\\
\hat d_i &=& \Gamma^d_i \theta_i, & i= 1,2,\ldots, N,
\ea\ee
where $G^d_i, \check{u}_i$ will be designed later. 
We remark that the internal model for  disturbance rejection is implemented also at agent 1 (the leader).\\
To both compensate the disturbance and achieve the desired formation, the proposed control law is composed of two parts: $\hat u$ guides the system to the desired formation, and $\hat d$ compensates the disturbance. Therefore, we can write the dynamics of each agent in the following compact form
\be\label{ud}
\begin{aligned}
u &= \hat u - \hat d\\
\dot \xi &= f(\xi) + g(\xi) (\hat u - \hat d +d).
\end{aligned}
\ee
Define the new variable $\tilde d_i= \hat d_i -d_i$ and $\tilde \theta_i =\theta_i -w_i^d$, then\\
\be
\begin{aligned}
\dot{\tilde \theta}_i &= \Phi_i^d \tilde\theta_i+G_i^d \check{u}_i\\
\tilde d_i &= \Gamma_i^d {\tilde \theta}_{i}. 
\end{aligned}
\ee
In compact form we write
\be\label{d.tilde}
\ba{rcl}
\dot{\tilde \theta} =
\left(\ba{c}
\dot {\tilde \theta}_1\\
\dot {\tilde \theta}_2\\
\ldots\\
\dot {\tilde \theta}_n
\ea\right)
=
\underbrace{
\left(\ba{cccc}
\Phi_1^d & \mathbf{0} & \ldots & \mathbf{0}\\
\mathbf{0} & \Phi_2^d & \ldots & \mathbf{0}\\
\vdots & \vdots & \ddots & \vdots \\
\mathbf{0} & \mathbf{0} & \ldots & \Phi_n^d\\
\ea\right)}_{\Phi^d}
\tilde \theta 
 + 
\underbrace{
\left(\ba{cccc}
G_1^d & \mathbf{0} & \ldots & \mathbf{0}\\
\mathbf{0} & G_2^d & \ldots & \mathbf{0}\\
\vdots & \vdots & \ddots & \vdots \\
\mathbf{0} & \mathbf{0} & \ldots & G_n^d\\
\ea\right)}_{G^d}
\underbrace{
\left(\ba{c}
\check u_1\\
\check u_2\\
\ldots\\
\check u_n
\ea\right)}_{\check u} \\\\
\tilde d =
\underbrace{\left(\ba{cccc}
\Gamma^d_1 & \mathbf{0} & \ldots & \mathbf{0}\\
\mathbf{0} & \Gamma^d_2 & \ldots & \mathbf{0}\\
\vdots & \vdots & \ddots & \vdots \\
\mathbf{0} & \mathbf{0} & \ldots & \Gamma^d_n\\
\ea\right)}_{\Gamma^d}\tilde \theta. \hspace{75mm}
\ea\ee
From the equation (\ref{ocls}) together with the equations (\ref{ud}), (\ref{d.tilde}) the overall closed-loop system is
\be\label{ocls.d}\ba{rcl} 
\dot {\tilde z} &=&  (B^T\otimes I_p) (h(\xi)+\tilde  \Gamma^v \tilde \eta)\\
\dot \xi&=& f(\xi) + g(\xi) (\hat u - \Gamma^d \tilde \theta)\\
\dot{\tilde \eta} &=& \tilde \Phi \tilde \eta +\tilde G \tilde u\\
\dot{\tilde \theta} &=& \Phi^d \tilde \theta +G^d \check u\\ 
\ea\ee
The following is proven:
\begin{proposition}\label{prop.dist.rej.+track}
Assume that $\Phi$ and $\Phi^d_i$, for $i=1,2,\ldots, N$ are skew-symmetric matrices. If 
\be\label{choice.d}\ba{rcllrcl}
G &=& (\Gamma^v)^T\\
G^d_i &=& (\Gamma^d_i)^T, 
\ea\ee
then all the Krasowskii solutions to (\ref{ocls.d}) in closed-loop with the control input 
\be\label{control.d}
\hat u=\tilde u = -(B\otimes I_p)\sign {\tilde z},\quad \check u = h(\xi)
\ee
are bounded and converge to the largest weakly invariant set for 
\be\label{tired}
\left(\ba{c}
\dot{\tilde z}\\
\mathbf{0}\\
\dot{\tilde \eta} \\
\dot{\tilde\theta}
\ea
\right)
\in
\left(\ba{c}
(B^T\otimes I_p) \tilde \Gamma^v \tilde \eta\\
-(B\otimes I_p) \mathcal{K}\; {\rm sign} \tilde z - \Gamma^d \tilde \theta \\
\tilde \Phi \tilde \eta -(\tilde \Gamma^v)^T (B\otimes I_p)\mathcal{K}\; \sign {\tilde z}\\
\Phi^d\tilde \theta
\ea\right),
\ee
%
%
such that $ \xi=\mathbf{0}$.
\end{proposition}

\begin{proof}
Consider the Lyapunov function
\[
V(\tilde z,\tilde v, \tilde \eta, \tilde \theta) = 
||\tilde z||_1+ S(\xi)+ \frac{1}{2} \tilde \eta^T \tilde \eta+ \frac{1}{2} \tilde \theta^T \tilde \theta.
\]
Any Krasowskii solution to (\ref{ocls.d}) with the control law defined by (\ref{control.d}) satisfies a differential inclusion of the form (\ref{eq:b7}) where $X= (\tilde z,\xi, \tilde \eta, \tilde \theta)$ and 
\begin{equation*}
F(\tilde z,\xi, \tilde \eta, \tilde \theta)=\left(\ba{c}
(B^T\otimes I_p) (h(\xi)+\tilde \Gamma ^v \tilde \eta) \\
f(\xi) - g(\xi) \Gamma^d \tilde \theta \\
\tilde \Phi\tilde \eta\\
\Phi^d \tilde \theta +G^d h(\xi)
\ea\right)- 
\left(\ba{c}
\mathbf{0}\\
g(\xi) (B\otimes I_p) \\
\tilde G (B\otimes I_p)\\
\mathbf{0}
\ea\right){\mathcal K}\,\sign\; \tilde z.
\end{equation*}
For any $w\in F(\tilde z, \xi, \tilde \eta, \tilde \theta)$, there exists $w^{\tilde z}\in {\mathcal K}\sign\; \tilde z$ such that
\[
w=
\left(\ba{c}
(B^T\otimes I_p) (h(\xi)+\tilde \Gamma ^v \tilde \eta) \\
f(\xi) - g(\xi) \Gamma^d \tilde \theta \\
\tilde \Phi\tilde \eta\\
\Phi^d \tilde \theta +G^d h(\xi)
\ea\right)-
\left(\ba{c}
\mathbf{0}\\
g(\xi) (B\otimes I_p) \\
\tilde G (B\otimes I_p)\\
\mathbf{0}\ea\right)w^{\tilde z}.
\]
Moreover for any $p\in \partial V(\tilde z, \xi, \tilde \eta, \tilde \theta)$ there exists $p^{\tilde z}\in \partial ||\tilde z||_1$, with 
\[
p^{\tilde z}_{k\ell}\in 
\left\{\ba{lll}
\{\sign\;{\tilde z}_{k\ell}\} & if & {\tilde z}_{k\ell}\ne 0\\
~[-1,+1] & if & {\tilde z}_{k\ell}= 0,
\ea\right.
\]
and $k=1,2,\ldots,m$, $\ell=1,2,\ldots, p$, such that
\[
p=\left(\ba{c}
p^{\tilde z}\\ \nabla S(\xi) \\ \tilde \eta\\ \tilde \theta
\ea\right).
\]
For $w^{\tilde z}\in {\mathcal K}\sign\; \tilde z$, let $p\in \partial V(\tilde z, \xi, \tilde \eta, \tilde \theta)$ be such that $w^{\tilde z}=p^{\tilde z}$. Hence,
\[
\ba{l}
\begin{aligned}
\langle w, p\rangle &= \langle (B^T\otimes I_p) (h(\xi)+\tilde \Gamma ^v \tilde \eta), w^{\tilde z}\rangle\\[2mm]
&+\langle f(\xi) - g(\xi) ((B\otimes I_p)w^{\tilde z} + \Gamma^d \tilde \theta), \nabla S(\xi)\rangle\\[2mm]
&+\langle \tilde \Phi\tilde \eta-\tilde G\; (B\otimes I_p)w^{\tilde z}, \tilde \eta\rangle + \langle \Phi^d \tilde \theta + G^d h(\xi), \tilde \theta \rangle.\\[2mm]
\end{aligned}
\ea
\]
By the condition (\ref{choice.d}), the above result simplifies as
\be\label{eq.11}
\langle w, p\rangle = \langle (B^T\otimes I_p) h(\xi), w^{\tilde z}\rangle 
+\langle f(\xi) - g(\xi) ((B\otimes I_p)w^{\tilde z} + \Gamma^d \tilde \theta), \nabla S(\xi)\rangle 
+\langle G^d h(\xi), \tilde \theta \rangle.\\[2mm]
\ee
Since the dynamics of each agent is strictly passive, we have
\begin{equation*}
\langle f(\xi) - g(\xi) ((B\otimes I_p)w^{\tilde z}+\Gamma^d \tilde \theta), \nabla S(\xi)\rangle \leq 
 -W(\xi)-\langle h(\xi), (B\otimes I_p)w^{\tilde z}+\Gamma^d \tilde \theta\rangle.
\end{equation*}
When it is replaced in (\ref{eq.11}), it gives $\langle w, p\rangle \leq -W(\xi)$, in view of $\Gamma^d=(G^d)^T$.
As in the previous proofs, this shows $\dot{\overline{V}}(\tilde z,\xi, \tilde \eta, \tilde \theta)=\{a\in \R: a\le -W(\xi)\}$ and hence boundedness of the solutions and their convergence to the largest weakly invariant set for the closed-loop system such that $\xi=\mathbf{0}$. 
This proves the thesis. 
\end{proof}

%

Proposition \ref{prop.dist.rej.+track} only proves boundedness of the solution and convergence of the velocity of each agent to its own reference velocity ($ \xi=\mathbf{0}$). To prove the convergence of the network to the desired formation and evolution of the network with the desired velocity, one should additionally prove that $\tilde z$ converges to the origin. {This result does not hold in general but there are a few cases of interest where this is true. In what follows we study two of these cases. 

{\it Case I: $\Phi=\mathbf{0}$, $\Phi_i^d=\mathbf{0}_{p\times p}$, $i=1,2,\ldots, N$, 
and $\Gamma^v$, $\Gamma_i^d$, $i=1,2,\ldots, N$, are nonsingular}.\\
This case correspond to the scenario in which the unknown reference velocity and the disturbances are constant signals. The arguments below show that the distributed control laws characterized in Proposition \ref{prop.dist.rej.+track} guarantee the achievement of the desired formation and the prescribed velocity while rejecting the disturbances. \\
From Proposition \ref{prop.dist.rej.+track} it is known that the system converges to the largest weakly invariant set for (\ref{tired}) such that $\xi=\mathbf{0}$. In the case $\Phi=\mathbf{0}$, $\Phi_i^d=\mathbf{0}_{p\times p}$, this implies that for any state $(\tilde z, \xi, \tilde \eta, \tilde \theta)=(\tilde z, \mathbf{0}, \tilde \eta, \tilde \theta)$ on this set, and for any  $w\in F(\tilde z, \xi, \tilde \eta, \tilde \theta)$, there exists $w^{\tilde z}\in {\mathcal K}\;\sign\; \tilde z$, such that
\be\label{caseI}\ba{rcl} 
\dot {\tilde z} &=& (B^T\otimes I_p) (\tilde \Gamma^v \tilde \eta)\\
\mathbf{0}&=& -(B \otimes I_p) w^{\tilde z} -\Gamma^d \tilde \theta\\
\dot{\tilde \eta} &=& -(\tilde \Gamma^v)^T (B \otimes I_p) w^{\tilde z}\\
\dot{\tilde \theta} &=& \mathbf{0}.\\ 
\ea\ee
From $\dot{\tilde \theta} = \mathbf{0}$, we conclude that $\tilde \theta$ is a constant vector. Hence, considering the second equality in (\ref{caseI}), we conclude that $ (B\otimes I_p) w^{\tilde z}$ is a constant vector. 
From the third equation we deduce that each component of $\dot {\tilde \eta}(t)$ is identically zero. As a matter of fact, the first component is zero by construction. The vectors $\dot{ \tilde \eta}_i(t)$, $i=2,\ldots,N$, are constant. If $\dot{ \tilde \eta}_i(t)$ is a non-zero vector for some $i$, then at least one of its components, say $\dot {\tilde \eta}_{i\ell}(t)$, must be non-zero and then $\tilde \eta_{i\ell}(t)$ would grow with constant velocity and this would contradict the boundedness of the solutions. Hence,  $\dot {\tilde \eta}(t)=\mathbf{0}$ and ${\tilde \eta}(t)$ is a  constant vector. From the first equation, the latter implies that  $\dot {\tilde z}(t)$ is constant as well. The same argument used for $\dot {\tilde \eta}(t)$ can be used again to show that $\dot {\tilde z}(t)=\mathbf{0}$. Now, consider the first equality in (\ref{caseI}) and multiply both sides by $B\otimes I_p$. Since $\dot {\tilde z} = \mathbf{0}$, we obtain $\mathbf{0} = (BB^T\otimes I_p) (\tilde \Gamma^v \tilde \eta)$. From the property of the Laplacian of a connected undirected graph, we conclude $\tilde \Gamma^v \tilde \eta = {\mathcal R}({\mathbf{1}_N\otimes I_p})$. 
From the definition of $\tilde \Gamma^v\tilde \eta$, we have $\tilde \Gamma^v\tilde \eta=\mathbf{0}$ and hence $\tilde \eta=\mathbf{0}$. 
The third identity in (\ref{caseI}), $\dot{\tilde\eta}=\mathbf{0}$, and the structure of $\tilde{\Gamma}^v$ implies that 
\[
(B\otimes I_p) w^{{\tilde z}}=
\left(\ba{c}
B_1^T\otimes I_p\\
B_2^T\otimes I_p
\ea
\right)w^{{\tilde z}}= \left(\ba{c}
c\\
\mathbf{0}
\ea
\right)
\]
for some constant $c\in \R^p$ (recall that it was proven previously that $(B\otimes I_p) w^{{\tilde z}}$ is a constant vector), and where $B_1^T\in \R^{1\times M}, B_2^T\in \R^{(N-1)\times M}$ is a partition of the incidence matrix $B$. We now prove that $c=\mathbf{0}$. Multiplying on the left by $\mathbf{1}_N^T\otimes I_p$ the previous expression, one obtains
\[
\mathbf{0}
= (\mathbf{1}_N^T\otimes I_p)(B\otimes I_p) w^{{\tilde z}}=
(\mathbf{1}_N^T\otimes I_p)\left(\ba{c}
c\\
\mathbf{0}
\ea
\right)=c
\]
as claimed. As a consequence, $(B\otimes I_p) w^{{\tilde z}}=\mathbf{0}$. With the same argument used in the proof of Proposition \ref{p1}, this implies that ${\tilde z}=\mathbf{0}$. Moreover, from the second equation in  (\ref{caseI}), $\Gamma^d \tilde \theta=\mathbf{0}$, and this shows that $\tilde \theta=\mathbf{0}$.

The arguments above allow us to infer the following:

\begin{corollary}
Assume that $\Phi=\mathbf{0}$ and $\Phi^d_i=\mathbf{0}$, for $i=1,2,\ldots, N$. Also assume that $\Gamma^v$, $\Gamma_i^d$, $i=2,3,\ldots, N$,  are non-singular and the graph $G$ is a connected undirected graph . If 
\be\label{choice}\ba{rcllrcl}
 \tilde G &=& (\tilde \Gamma^v)^T\\
G^d &=& (\Gamma^d)^T, 
\ea\ee
then all the Krasowskii solutions to the  system (\ref{ocls.d}) in closed-loop with the control input 
\be\label{control}
\hat u=\tilde u = -(B\otimes I_p)\sign\,{\tilde z},\quad \check u = h(\xi)
\ee
are bounded and converge to the 
origin.
\end{corollary}

{\em Case II: Harmonic disturbance rejection with known reference velocity.}\\
We consider now the case in which the reference velocity is known and the controller only adopts an internal model to reject the disturbances. In this case, the system (\ref{ocls.d}) becomes
\be\label{ocls.d.IV}\ba{rcl}
\dot {\tilde z} &=& (B^T\otimes I_p) h(\xi)\\
\dot{\xi} &=& f(\xi) + g(\xi) (\hat u - \Gamma^d \tilde \theta)\\
\dot{\tilde \theta} &=& \Phi^d \tilde \theta +G^d \check u.\\ 
\ea\ee
The choice $G^d =(\Gamma^d)^T$, $\hat u = -(B\otimes I_p) \sign{\tilde z}$, $\check u = h(\xi)$ as designed in Proposition \ref{prop.dist.rej.+track} yields that all the 
Krasowskii solutions to the closed-loop system 
\be\label{ocls.d.IV.cl}\ba{rcl}
\dot {\tilde z} &=& (B^T\otimes I_p) h(\xi)\\
\dot{\xi} &=& f(\xi) + g(\xi) (-(B\otimes I_p) \sign{\tilde z}- \Gamma^d \tilde \theta)\\
\dot{\tilde \theta} &=& \Phi^d \tilde \theta +(\Gamma^d)^T h(\xi)\\ 
\ea\ee
are bounded and converge to the largest weakly invariant set for (\ref{ocls.d.IV.cl}). In the case 
the graph ${\mathcal G}$ has no loops the following holds:

\begin{corollary}
If the graph ${\mathcal G}$ has no loops and for $i=1,2\ldots, N$ the exosystems have matrices $(\Gamma_i^d, \Phi^d_i)$ of the form
\[
\Phi^d_{i}={\rm block.diag} \left\{
\left(\ba{cc} 
0 & \omega_{i1}\\
-\omega_{i1}& 0\\
 \ea \right),\ldots, \left(\ba{cc} 
0 & \omega_{ip}\\
-\omega_{ip}& 0\\
 \ea \right)
\right\},
\]
with $\omega_{i\ell}\ne 0$ for all $\ell=1,\ldots,p$, and 
\[
\Gamma_i^d= 
{\rm block.diag} \left\{
\Gamma_{i1}^d, \ldots, \Gamma_{ip}^d
\right\}
\]
with $\Gamma_{i\ell}^d\ne \mathbf{0}$, for all $\ell=1,\ldots,p$, 
then all the Krasowkii solutions to (\ref{ocls.d.IV.cl}) converge to the set of points where $\xi=\mathbf{0}$ and $\tilde z=\mathbf{0}$.  
\end{corollary}
\begin{proof}
A solution on the weakly invariant set for (\ref{ocls.d.IV.cl}), where $\xi=\mathbf{0}$,  is such that $(\tilde z, \tilde \theta)$ satisfies 
\be\label{dyn.inv.set}\ba{rcl}
\dot{\tilde z} &=& \mathbf{0}\\
\dot{\tilde\theta} &=& \Phi^d\tilde \theta
\ea\ee
and there exists $w^{\tilde z}\in  \mathcal{K}\; {\rm sign} \tilde z$ such that 
$\mathbf{0}= g(\mathbf{0})(-(B\otimes I_p) w^{\tilde z}- \Gamma^d \tilde \theta)$. 
Since by assumption $g(\mathbf{0})$ is full-column rank, the latter is equivalent to
\be\label{critical1}
\mathbf{0}= -(B\otimes I_p) w^{\tilde z}- \Gamma^d \tilde \theta. 
\ee
If the graph ${\mathcal G}$ has no loops, the edge Laplacian matrix $B^T B$ is non-singular (\cite{mesbahi.egerstedt.book}). Then, from (\ref{critical1}) we obtain
\be\label{critical2}
w^{\tilde z}= - (B^T B \otimes I_p)^{-1}\ (B^T \otimes I_p)\ \Gamma^d \tilde \theta. 
\ee
From (\ref{dyn.inv.set}), $\dot{\tilde\theta} = \Phi^d\tilde \theta$ implies that each $\Gamma^d_{i\ell}\tilde \theta_{i\ell}$ is a harmonic signal. One can write
\begin{equation*}
\Gamma^d_{i\ell}\tilde \theta_{i\ell} = \alpha_{i\ell}\cos (\omega_{i\ell}\;t + \varphi_{i\ell})
\end{equation*}
where $\alpha_{i\ell}, \varphi_{i\ell}$ are constants,
and $t \geq t_c$ ($t_c$ is the time at which the system converges to the invariant set).
 Therefore, (\ref{critical2}) implies that each $w^{{\tilde z}}_{k\ell}$ is a linear combination of harmonic signals. One obtains
\be\label{critical3}
w^{{\tilde z}}_{k\ell} = -\sum_{i,\ell} \beta_{i\ell}\cos (\omega_{i\ell}\;t + \varphi_{i\ell})
\ee
for some constants $\beta_{i\ell}$ (where the dependence on $k$ was neglected). We claim that (\ref{critical3}) implies that ${\tilde z}_{k\ell}$ is equal to zero. By contradiction, assume that ${\tilde z}_{k\ell}$ is not zero. Since $\dot{\tilde z}_{k\ell}=0$ (from (\ref{dyn.inv.set})), then ${\tilde z}_{k\ell} = c$, where $c \in \R$ is a non-zero constant. As a result, $w^{{\tilde z}}_{k\ell}$ is either $-1$ or $+1$. Without loss of generality, let us assume $w^{{\tilde z}}_{k\ell} = 1$. Now, multiply (\ref{critical3}) by $(1-\frac{t}{\lambda})$, for some positive $\lambda$, and integrate 
both sides of the equality. For the left-hand side when $w^{{\tilde z}}_{k\ell} = 1$ one obtains
\be\label{ce2}
\int_{0}^{\lambda}\ (1-\frac{t}{\lambda}) dt= 
\frac{\lambda}{2}.
\ee
On the other hand, the right-hand side of (\ref{critical3}) yields
\begin{equation}\label{ce3}
\begin{aligned}
\int_{0}^{\lambda}\ (1-\frac{t}{\lambda})\ \sum_{i, \ell} \beta_{i\ell}\cos(\omega_{i\ell}\;t + \varphi_{i\ell}) dt=\\ 
\sum_{i, \ell} 
\int_{0}^{\lambda} \beta_{i\ell}\ (1-\frac{t}{\lambda})\cos(\omega_{i\ell}\;t + \varphi_{i\ell}) 
dt=\\
- \sum_{i, \ell} 
\beta_{i\ell}\ (\frac{\sin(\varphi_{i\ell})}{\omega_{i\ell}} - \frac{\cos(\lambda \omega_{i\ell} + \varphi_{i\ell})}{\lambda \omega_{i\ell}^2} + 
\frac{cos(\varphi_{i\ell})}{\lambda {\omega_{i\ell}}^2}) 
\end{aligned}
\end{equation}
Since $\omega_{i\ell} \neq 0$, $N, p$ are finite numbers, the above sum has a bounded value. As $\lambda$ can be any positive number, from (\ref{ce2}) and (\ref{ce3}), it can be concluded that the equality (\ref{critical3}) under the assumption ${\tilde z}_{k\ell} = c$ cannot hold. Therefore, 
${\tilde z}_{k\ell} = 0$. Since the same holds true for any $k,\ell$, then $\tilde z= \mathbf{0}$.
\end{proof}

\begin{remark}
The condition on the graph to have no cycles is not necessary and one can find alternative statements that do not require the graph to be a tree but introduce conditions on the frequencies $\omega_{i\ell}$. 
\end{remark}

\section{A different controller design for velocity tracking and disturbance rejection}\label{sec5}
In the previous section, the distributed controllers that allow to keep the formation while tracking an unknown reference trajectory and rejecting disturbances apply only to some special cases. To overcome this limitation, we propose in this section a slightly different controller design. The alternative design is carried out for a more restricted class of systems. As a matter of fact, we assume that the dynamics (\ref{eq:a5}) of the agents satisfies
\[
g_i(\xi_i)= g_i,\quad y_i=h(\xi_i)=\xi_i, 
\]
that is the input vector field is constant (and has full-column rank) and the passive output $y_i$ coincides with the state $\xi_i$. Finally, we ask the function $W_i(\xi_i)$ in (\ref{eq:passivity})  to be lower bounded by a quadratic term, namely   $W_i(\xi_i)\ge ||\xi_i||^2 $. 

The main difference in the design lies in the controller (\ref{im.d}) introduced to counteract the effect of the disturbances. Although we still adopt the same structure, namely 
\be\label{im.d2} 
\ba{rcll}
\dot{\theta}_i &=& \Phi^d_i \theta_i +G^d_i \check{u}_i\\
\hat d_i &=& \Gamma^d_i \theta_i, & i= 1,2,\ldots, N,
\ea\ee
we give $G^d_i \check{u}_i$ a different interpretation, namely we let 
\[
\check{u}_i = y_i-\hat \xi_i,
\]
where $\hat \xi_i$ is an additional state of the controller that obeys the equation
\[\ba{rcl}
\dot {\hat \xi}_i &=& f_{i} (y_{i} ) + g_{i} (u_{i}+\Gamma_i^d  \theta_i) + H_{i}(\xi_i-\hat\xi_i),
\ea\]
and $G^d_{i},H_{i}$ are two gain matrices. The variables
\[
\tilde \xi_i =\xi_i-\hat \xi_i, \quad  \tilde\theta_i =\theta_i-w^d_i
\]
obey the equations
\be\label{est.err.var}\ba{rcl}
\dot{{\tilde \xi}}_i &=& - g_i \Gamma_i^d \tilde \theta_i -H_{i}\tilde \xi_i \\
\dot{{\tilde \theta}}_i &=& \Phi_i^d \tilde \theta_i  + G_{i}^d\tilde \xi_i \\
\ea\ee
The pair
\[
\left(\ba{cc}
I & \mathbf{0}
\ea\right), \quad 
\left(\ba{cc}
 \mathbf{0} & - g_i \Gamma_i^d \\
 \mathbf{0} & \Phi_i^d 
\ea\right)
\]
is observable provided that $g_i$ is full-column rank and the pair $(\Gamma_i^d, \Phi_i^d)$ is observable. Under these conditions, there exist matrix gains $G_{i}^d, H_{i}$  such that 
the estimation error system (\ref{est.err.var}) converges exponentially to zero. Let $P_i=P_i^T>0$ be a matrix such that 
\[
\left(\ba{cc}
-H_i & -g_i \Gamma_i^d \\
G_i^d  & \Phi_i^d 
\ea\right)^T P_i +P_i 
\left(\ba{cc}
-H_i & -g_i \Gamma_i^d \\
G_i^d  & \Phi_i^d 
\ea\right)=
-2
\left(\ba{cc}
I &  \mathbf{0} \\
 \mathbf{0} & \gamma I
\ea\right)
\]
where $\gamma$ is a constant satisfying $\gamma\ge \max\{||{\Gamma_i^d}^T{\Gamma_i^d}||,\, 1\le i\le N \}$.\\
Consider the overall system
\be\label{ocls.d2}\ba{rcl} 
\dot {\tilde z} &=& (B^T\otimes I_p) (\xi+\tilde  \Gamma^v \tilde \eta)\\
\dot \xi&=& f(\xi) + g(\hat u - \Gamma^d \tilde \theta)\\
\dot{\tilde \eta} &=& \tilde \Phi \tilde \eta +\tilde G \tilde u\\
\dot{{\tilde \theta}} &=& \Phi^d \tilde \theta +G^d\tilde \xi \\
\dot{{\tilde \xi}} &=& -g \Gamma^d \tilde \theta -H\tilde \xi \\
\ea\ee
where $u+\Gamma^d w^d=\hat u-\Gamma^d \theta+\Gamma^d w^d=  \hat u - \Gamma^d \tilde \theta$, $H={\rm block.diag}\{H_1,H_2, \ldots, H_N\}$, and design the control inputs as
\[
\hat u= \tilde u =-(B\otimes I_p) {\rm sign}{\tilde z},
\] 
so that the closed-loop system becomes
\be\label{ocls.d2.cl}\ba{rcl} 
\dot {\tilde z} &=&  (B^T\otimes I_p) (\xi+\tilde  \Gamma^v \tilde \eta)\\
\dot \xi&=& f(\xi) + g(-(B\otimes I_p) {\rm sign}{\tilde z} - \Gamma^d \tilde \theta)\\
\dot{\tilde \eta} &=& \tilde \Phi \tilde \eta -\tilde G (B\otimes I_p) {\rm sign}{\tilde z}\\
\dot{{\tilde \theta}} &=& \Phi^d \tilde \theta  +G^d\tilde \xi \\
\dot{{\tilde \xi}} &=& -g \Gamma^d \tilde \theta -H\tilde \xi \\
\ea\ee

\begin{proposition}
Assume that the pairs $(\Gamma^v, \Phi)$ and $(\Gamma^d_i, \Phi_i^d)$, for $i=1,2,\ldots,N$, are observable. Let $G=(\Gamma^v)^T$. Then,
all the Krasowskii solutions to the closed-loop system (\ref{ocls.d2.cl}) converge to the origin. 
\end{proposition}

\begin{proof}
The analysis follows a similar trail as for the other proofs, with a few significant variations. We focus on the set-valued derivative 
\begin{equation*}
\dot{\overline V}(\tilde z,\xi, \tilde \eta, \tilde \theta, \tilde \xi)= \{a\in \R\,:\, \exists w\in F(\tilde z,\xi, \tilde \eta, \tilde \theta, \tilde \xi)\ \ {\rm s.t.}\ \ 
a=\langle w, p \rangle,\  \mbox{for all} \ \ p\in \partial V (\tilde z,\xi, \tilde \eta, \tilde \theta, \tilde \xi)\},
\end{equation*}
In this case the Lyapunov function is 
\[\ba{rl}
V(\tilde z,\xi, \tilde \eta, \tilde \theta, \tilde \xi) = 
||\tilde z||_1+ S(\xi)+ \frac{1}{2} \tilde \eta^T \tilde \eta+ 
 \dst\frac{1}{2}
\dst\sum_{i=1}^N (\tilde \theta_i^T\; \tilde \xi_i^T) P_i (\tilde \theta_i^T\; \tilde \xi_i^T)^T,
\ea\]
where $P={\rm block.diag}\{P_1,\ldots, P_N\}$ and $F$ is obtained from
the right-hand side of (\ref{ocls.d2.cl}), namely
\[
\left(\ba{c}
(B^T\otimes I_p) (\xi+\tilde  \Gamma^v \tilde \eta)\\
f(\xi) - g\Gamma^d \tilde \theta\\
\tilde \Phi \tilde \eta\\
\Phi^d \tilde \theta  +G^d\tilde \xi\\
-g \Gamma^d \tilde \theta -H\tilde \xi
\ea\right)
+
\left(\ba{c}
\mathbf{0}\\
-g(B\otimes I_p)\\
-\tilde G (B\otimes I_p)\\
\mathbf{0}\\
\mathbf{0}\\
\ea\right)
\mathcal{K}\,{\rm sign}\tilde z.
\]
Following the line of arguments already used in this paper, for each state such that  $\dot{\overline V}$ is a non-empty set, any element of  $\dot{\overline V}$ will satisfy
\[
\langle p,v\rangle \le -W(\xi)- \xi^T \Gamma^d \tilde \theta- ||\tilde \xi||^2 -\gamma ||\tilde \theta||^2.
\]
By the condition on $W(\xi)$ and the definition of $\gamma$, a completion of the squares argument implies that $\langle p,v\rangle \le -\frac{|| \xi ||^2}{2}-\gamma \frac{|| \tilde\theta ||^2}{2}- ||\tilde \xi||^2$. As a result, for all the states for which $\dot{\overline V}\ne \emptyset$, we have $\dot{\overline V}=\{a\in \R: a\le -\frac{|| \xi ||^2}{2}-\gamma \frac{|| \tilde\theta ||^2}{2}- ||\tilde \xi||^2\}\subseteq (-\infty, 0]$. This implies convergence to the largest weakly invariant set where $\xi=\mathbf{0}$, $\tilde\theta=\mathbf{0}$, $\tilde \xi=\mathbf{0}$. On this invariant set, system (\ref{ocls.d2.cl}) becomes
\[\ba{rcl} 
\dot {\tilde z} &=&  (B^T\otimes I_p) \tilde  \Gamma^v \tilde \eta\\
\mathbf{0} &=& -(B\otimes I_p) {\rm sign}{\tilde z}\\
\dot{\tilde \eta} &=& \tilde \Phi \tilde \eta -\tilde G (B\otimes I_p) {\rm sign}{\tilde z}\\
\ea\]
The identity $\mathbf{0} = -(B\otimes I_p) {\rm sign}{\tilde z}$ implies $\tilde z=\mathbf{0}$ and the equations above can be further simplified as 
\[\ba{rcl} 
\mathbf{0} &=&  (B^T\otimes I_p) \tilde  \Gamma^v \tilde \eta\\
\dot{\tilde \eta} &=& \tilde \Phi \tilde \eta.
\ea\]
Hence, $\tilde  \Gamma^v \tilde \eta=\mathbf{0}$ and by observability of $(\Gamma^v, \Phi)$, it follows that $\tilde \eta=\mathbf{0}$.  
\end{proof}

The result can be rephrased as follows. For the formation of agents 
\[\ba{rcl}
\dot x_i &=& y_i + v_i^r\\
\dot \xi_i &=& f_i(\xi_i) +g_i (u_i+d_i)\\
y_i &=& \xi_i
\ea\]
with prescribed reference velocity $v^*$ generated by (\ref{exosv}) and with disturbances $d_i$ generated by (\ref{exos.d}), the distributed controller
\[
\ba{rcl}
\dot{\eta} &=& \tilde \Phi \eta -G (B\otimes I_p){\rm sign} \tilde z\\
\dot{\theta} &=& \Phi^d \theta +G^d (y-\hat \xi)\\
\dot {\hat \xi} &=& f (y) + g (-(B\otimes I_p){\rm sign} \tilde z+\Gamma^d \theta) + H(y-\hat\xi)\\
u &=& -(B\otimes I_p){\rm sign} \tilde z-\Gamma^d \theta\\
\ea
\]
(where the first component of $v^r$ is $v^*$)
guarantees boundedness of all the states, and convergence to the desired formation ($\tilde z=\mathbf{0}$) with asymptotic tracking of the reference velocity ($\xi=\mathbf{0}$, $\tilde\eta=\mathbf{0}$) as well as disturbance rejection. 

\section{Simulations}\label{sec:sim}
\begin{figure}[t]\centering
\includegraphics [width=0.4\linewidth]{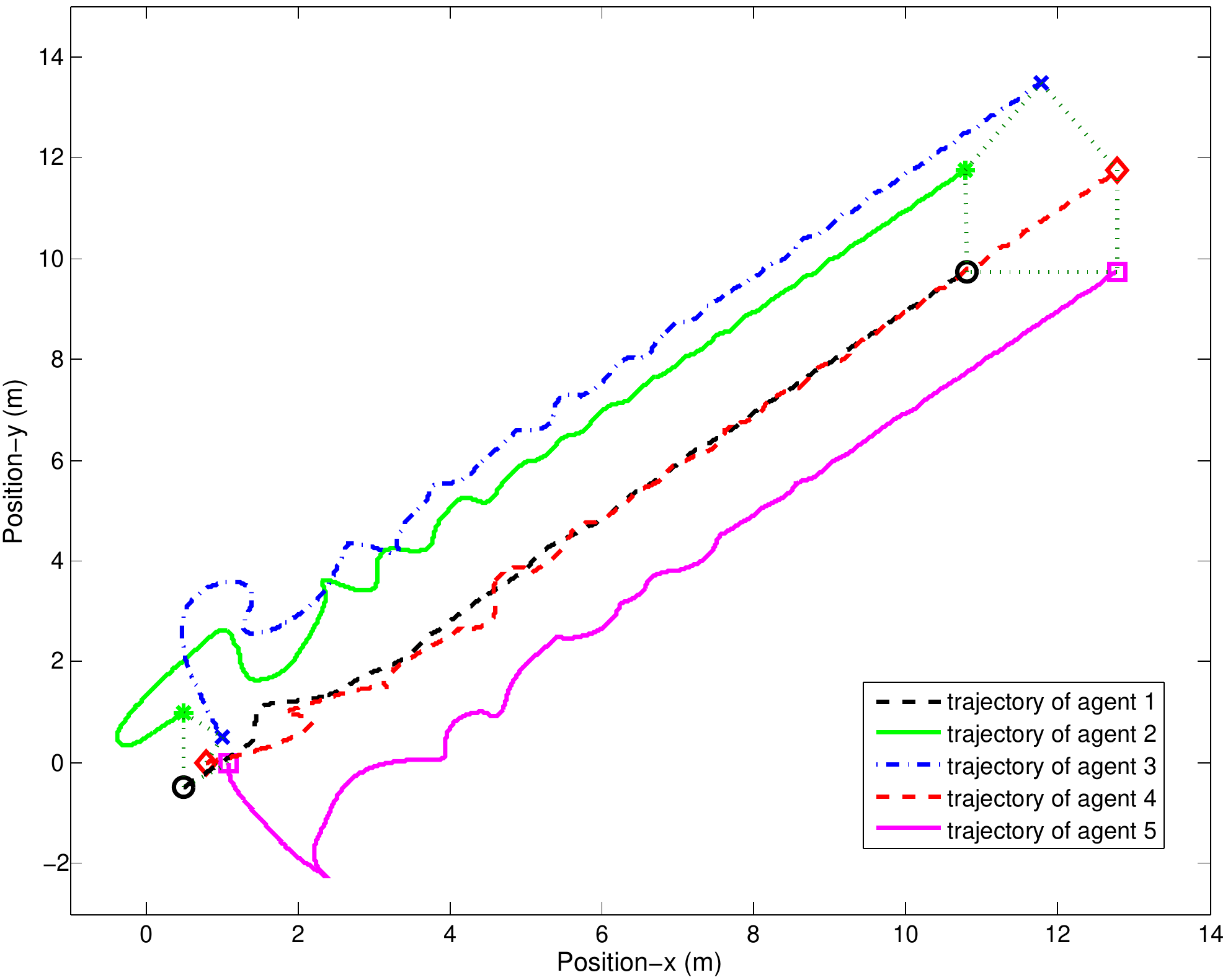}
\caption{The evolution of five agents in $\R^{2}$ tracking a constant reference velocity only known to the leader. The system reaches the desired formation and exhibits a translational motion with the desired constant velocity.}
\label{fig:f2}
\end{figure}
In this section we present the simulation results for a group of five strictly passive systems in $\R^{2}$. The dynamics of each agent is given by 
\begin{equation}
\mathcal{H}_{i}: \left\{
 \begin{array}{c}
  \dot{\xi_{i}} = - \xi_{i} + u_{i}\\
  y_{i} = \xi_{i}\hspace{12mm}\\
 \end{array}
\right.
\end{equation}
where, comparing with (\ref{eq:a5}), $f_{i} ( \xi_{i} )= -\xi_{i}$, $g_{i} ( \xi_{i} ) = I_2$, and $h_{i} ( \xi_{i} )=\xi_{i}$. The agents exchange information over a connected graph. The associated incidence matrix is
\[ B = \left( \begin{array}{ccccccc}
-1 & 0 & 0 & 0 & 0 & 0 \\
+1 & -1 & -1 & 0 & 0 & 0\\
0 & +1 & 0 & -1 & -1 & 0\\
0 & 0 & +1 & +1 & 0 & -1\\
0 & 0 & 0 & 0 & +1 & +1\end{array} \right).\]
The desired formation has a pentagonal shape with edge length equal to $2$ and is defined by the following inter-agent distance vectors: $z_{1}^{\ast} = [0\;\;2]^{T}$, $z_{2}^{\ast} = [1\;\;\sqrt{3}]^{T}$, $z_{3}^{\ast} = [2\;\;0]^{T}$, $z_{4}^{\ast} = [1\;\;\sqrt{3}]^{T}$, $z_{5}^{\ast} = [1\;\;-2-\sqrt{3}]^{T}$, $z_{6}^{\ast} = [0\;\;-2]^{T}$. Note that the number of edges of the graph is six. The initial position of the agent is set to $x(0)=[0.5\;\;-0.5\;\;0.5\;\;1\;\;1\;\;0.5\;\;0.8\;\;0\;\;1.1\;\;0]^{T}$. Figures 1 to 3 refer to the case of a formation evolving with a constant reference velocity known only to the formation leader (agent $1$). The following are set as internal model parameters for all agents: $\Phi=\mathbf{0}_{2\times2}$,\;$\Gamma ^v=I_2$,\;$G=I_2$,\;$w^v(0)=[1\ 1]^{T}$, and $\eta_{i}(0)=[0\ 0]^{T}$, for $i=2,\ldots,5$. Figure \ref{fig:f2} shows the evolution of the described system with the desired constant velocity $v^\ast = [1\ 1]^{T}$. The other agents generate the reference velocity using the control laws based on the internal model principle.

Figure \ref{fig:f3} shows the time behavior of the horizontal component of $\tilde{z}_{1}$, $\sign\,\tilde{z}_{1}$ and the corresponding control $\tilde{u}_{1}$. As time elapses, $\tilde{z}_{1}$ converges to the origin implying convergence to the desired relative position. While $\tilde{z}_{1}$ converges to the origin, $\sign\; \tilde{z}_{1}$ and $\tilde{u}_{1}$ converge to the discontinuity surface and oscillate between $+1$ and $-1$. The state variables associated with the other agents exhibit a similar behavior and are not shown. Figure \ref{fig:f5} shows the horizontal and vertical components of the reference velocity $v^\ast$ and the estimated velocities $v_i^r$. The leader generates a constant desired velocity and the follower agents estimate the same reference velocity after some time.\\

\begin{figure}[t]
\centering
\begin{minipage}{.4\textwidth}
\centering
\includegraphics [width=0.9\linewidth]{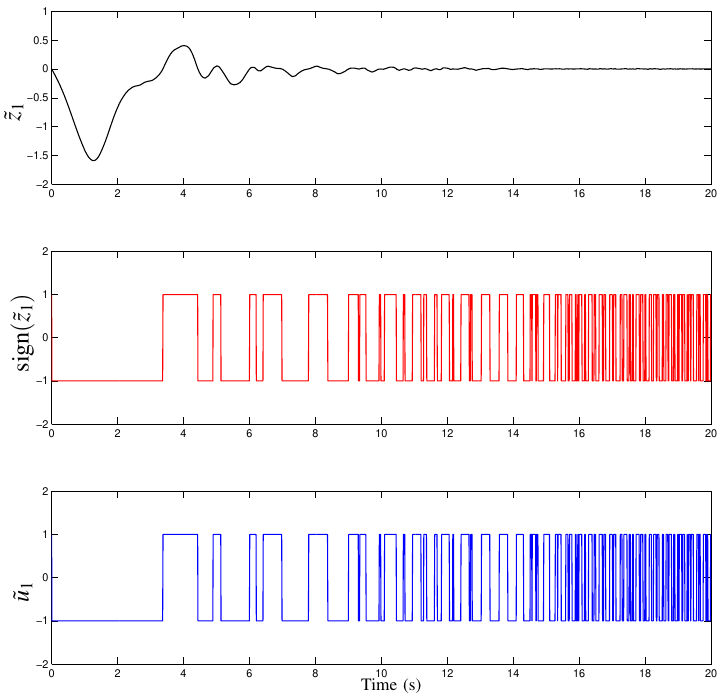}
\caption{The plot of the (horizontal) x-component of $\tilde{z}_{1}$, $\sign\; \tilde{z}_{1}$ and the corresponding control $\tilde{u}_{1}$.} 
\label{fig:f3}
\end{minipage}
\hfill
\begin{minipage}{.4\textwidth}
\centering
\includegraphics [width=0.9\linewidth]{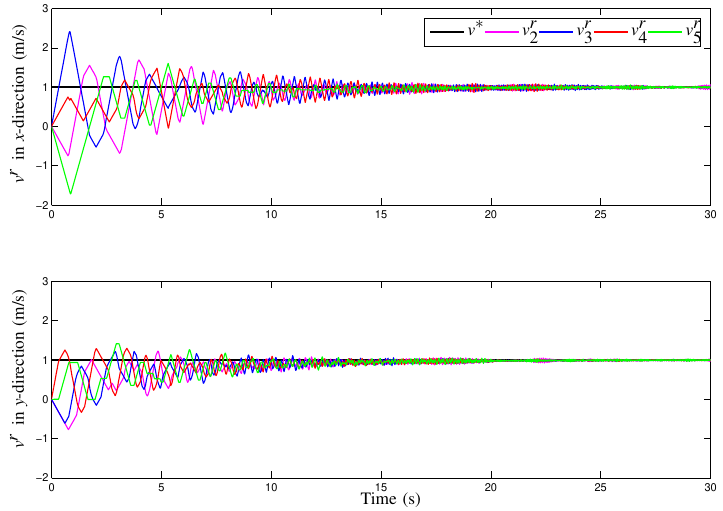}
\caption{The horizontal and vertical component of the estimated reference velocity $v_i^r$ for each agent $i$. The leader (agent 1) generates the desired velocity $v^{\ast}=[1\ 1]^{T}$(m/s).}
\label{fig:f5}
\end{minipage}
\end{figure}

Simulation results in the presence of matched disturbances are presented next. Figure \ref{fig:f7} shows the state $(\tilde z,\xi,\tilde \eta, \tilde \theta)$ of the system with constant disturbance and constant reference velocity (Case I). The reference velocity is only known to the agent one. The following are the parameters chosen for the simulation: $\Phi_i^d=\Phi=\mathbf{0}_{2\times2}$,\;$\Gamma ^v=\Gamma _i^d=I_2$,\;$G=G_i^d=I_2$,\; $w_i^d(0)=[i\ i+3]^{T}$ for $i=1 \hdots 5$ and $w^v(0)=[1\ 1]^{T}$, $\eta_{i}(0)=[0\ 0]^{T}$ for $i=2 \hdots 5$ .
As predicted, the formation is achieved, the desired velocity is reached by all the agents and the  disturbances are rejected. Figure \ref{fig:f8} shows the time behavior of the horizontal component of $\tilde{z}_{1}$, $\sign\,\tilde{z}_{1}$ and the corresponding control $\tilde{u}_{1}$. As time evolves, $\tilde{z}_{1}$ converges to the origin, and $\sign\; \tilde{z}_{1}$ and $\tilde{u}_{1}$ start switching between $+1$ and $-1$ with high frequency.
\begin{figure}[t]
\centering
\begin{minipage}{.5\textwidth}
\centering\includegraphics [width=0.9\linewidth]{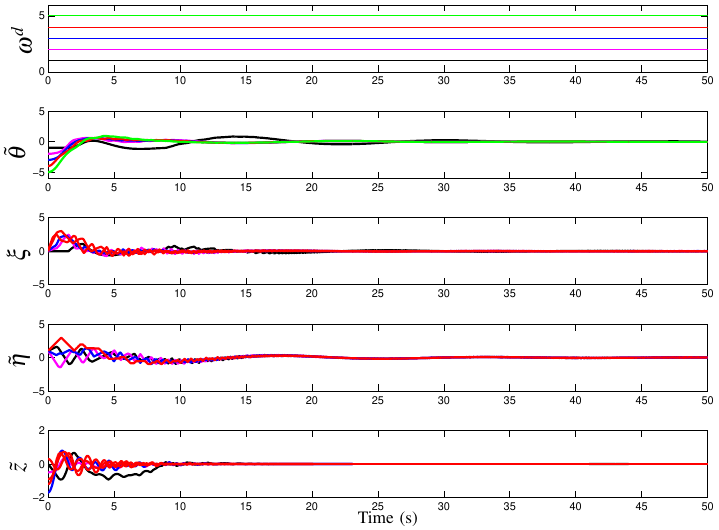}
\caption{The horizontal component of the disturbance $w^d$, the disturbance estimation error $\tilde \theta$ and the velocity error $\xi$ (for all 5 agents), together with the reference velocity estimation error $\tilde \eta$ (for follower agents) and the horizontal component of the relative position vector $\tilde z$. Here, the constant disturbance is rejected by the internal-model-based controller and $\xi$, $\tilde \eta$, $\tilde \theta$, $\tilde z$ converge to zero.}
\label{fig:f7}
\end{minipage}
\hfill
\begin{minipage}{.4\textwidth}
\centering
\includegraphics [width=0.9\linewidth]{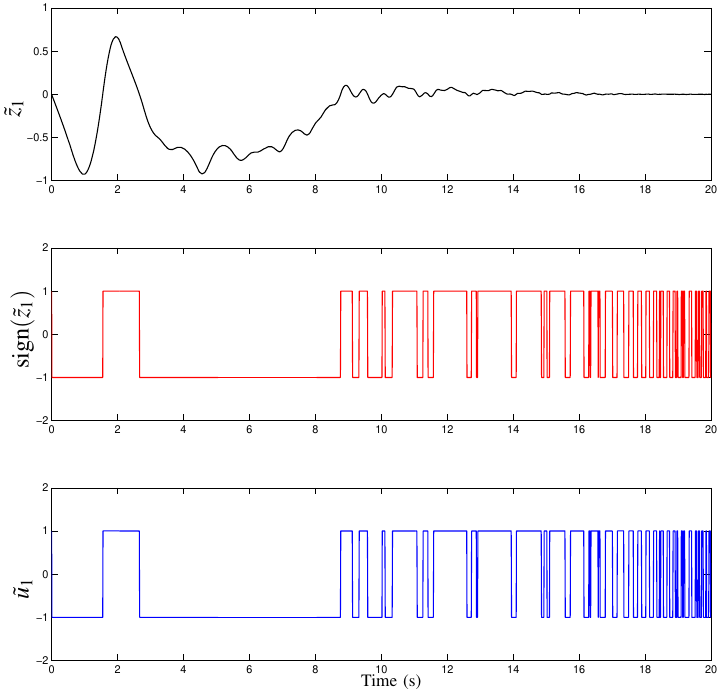}
\caption{The plot of the (horizontal) x-component of $\tilde{z}_{1}$, $\sign\; \tilde{z}_{1}$ and the corresponding control $\tilde{u}_{1}$ in the presence of constant reference velocity and constant disturbance.}
\label{fig:f8}
\end{minipage}
\end{figure}
Figures \ref{fig:f10} and \ref{fig:f11} show the state of the system with harmonic  disturbance and known reference velocity (Case II). The graph is considered to be a tree with an incidence matrix $B'$ which is obtained by removing the two last columns of the proposed matrix $B$. The desired reference velocity is $v^{\ast}=[1\ \ 1]^{T}$ and it is known to all of the agents. In this example, the following are set as the internal model parameters:\\
$ \Phi^d_1 = I_2 \otimes \left( \begin{array}{cc}
0 & 1 \\
-1 & 0 \end{array} \right)$,\; $\Gamma_i^d=\left( \begin{array}{cccc}
0.5 & 0.5 &0 & 0 \\
0 & 0 & -0.5 & 0.5 \end{array} \right)$,\;
$G_i^d={\Gamma _i^d}^{T}$,\;$w_i^d(0)=[0.1\ 0.1\ 0.1\ 0.1]^{T}$, 
$ \Phi^d_2 = 2\Phi^d_1$,
 $\Phi^d_1=\Phi^d_3=\Phi^d_5$, $\Phi^d_4= 0.5 \Phi^d_1.$
Figure \ref{fig:f10} shows the results when the disturbance is tackled by a controller based on the internal model principle. The result confirms that both the desired formation and the desired velocity are achieved. Figure \ref{fig:f11} shows the time behavior of the horizontal component of $\tilde{z}_{1}$, $\sign\,\tilde{z}_{1}$ and the corresponding control $\tilde{u}_{1}$. Similar to the case without the disturbance, while $\tilde{z}_{1}$ converges to the origin, $\sign\; \tilde{z}_{1}$ and $\tilde{u}_{1}$ converge to the discontinuity surface and oscillate between $+1$ and $-1$.

\begin{figure}[htpb]
\centering
\begin{minipage}{.5\textwidth}
\includegraphics [width=0.9\linewidth]{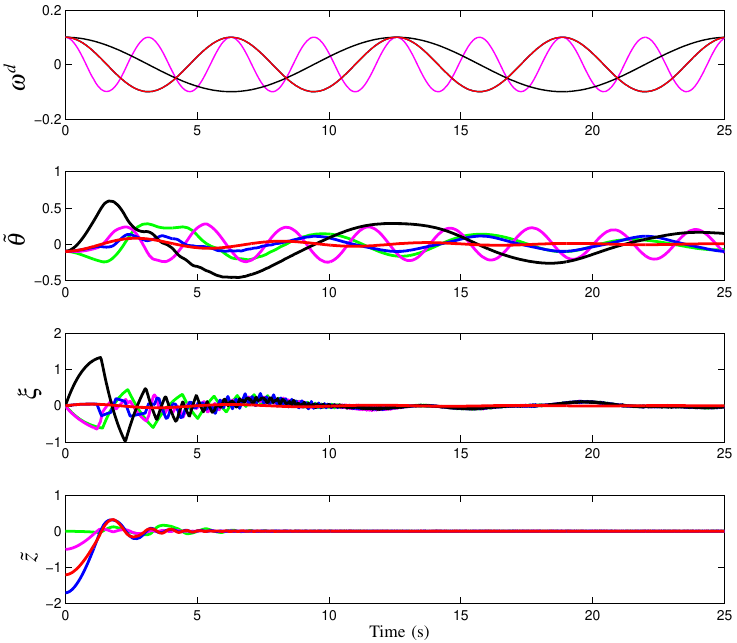}
\caption{The horizontal component of the disturbance $w^d$, the disturbance estimation error $\tilde \theta$ and the velocity error $\xi$ for all five agents, together with the horizontal component of the relative position vector $\tilde z$. Here, an internal-model-based controller is used to reject the harmonic disturbance. $\xi$ and $\tilde z$ converge to zero.}
\label{fig:f10}
\end{minipage}
\hfill
\begin{minipage}{.4\textwidth}
\centering
\includegraphics [width=0.9\linewidth]{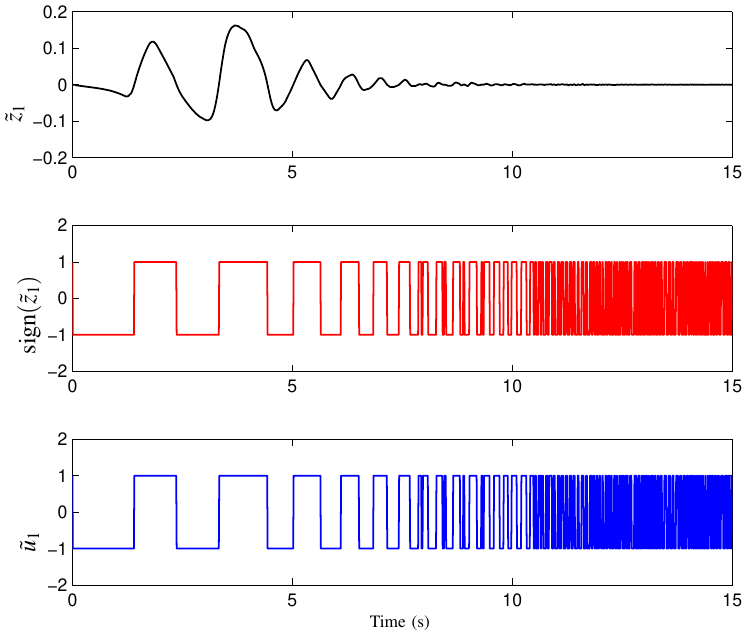}
\caption{The plot of the (horizontal) x-component of $\tilde{z}_{1}$, $\sign\; \tilde{z}_{1}$ and the corresponding control $\tilde{u}_{1}$ in the presence of harmonic disturbance.} 
\label{fig:f11}
\end{minipage}
\end{figure}

Simulation results concerning the controllers proposed in Section \ref{sec5} are shown in Figures \ref{fig:f12} and \ref{fig:f13}. The disturbance is assumed to be a linear combination of a constant and a harmonic signal 
for all five agents. The desired constant reference velocity, $v^{\ast}=[1\ \ 1]^{T}$, is only known to the formation leader (agent 1). In this example, the following are set as the parameters of the  disturbance and the controller:\\
$ \Phi^d_i = I_2 \otimes \left( \begin{array}{ccc}
0 & 0 & 0\\
0 & 0 & 2\\
0 & -2 & 0\end{array} \right)$,\; $\Gamma_i^d=\left( \begin{array}{cccccc}
0.5 & 0.5 & 0.5 &0 & 0 & 0\\
0 & 0 & 0 & 0.5 &-0.5 & 0.5 \end{array} \right)$,\;
$G_i^d={\Gamma _i^d}^{T}$,\;$w_i^d(0)=[0.1\ 0.1\ 0.1\ 0.1\ 0.1\ 0.1]^{T}$, for $i=1,\ldots,5$. Also, 
$\Phi=\mathbf{0}_{2\times2}$,\;$\Gamma ^v=I_2$,\;$G=I_2$,\;$w^v(0)=[1\ 1]^{T}$, and $\eta_{i}(0)=[0\ 0]^{T}$, for $i=2,\ldots,5$. 
The observer gain $H_i$ is set to $50\ I_2$, and the velocity error dynamics is $f_i(\xi_i)= -30\ I_2\; \xi_i$ and $g_i=10\ I_2$.
Figure \ref{fig:f12} shows the results when the disturbance is tackled by a controller based on the designed described in Section \ref{sec5}. The result confirms that the desired formation and the desired velocity are attained. In addition, the disturbance is rejected. Figure \ref{fig:f13} shows the time behavior of the horizontal component of $\tilde{z}_{1}$, $\sign\,\tilde{z}_{1}$ and the corresponding control $\tilde{u}_{1}$. Similar to the previous cases, while $\tilde{z}_{1}$ converges to the origin, $\sign\; \tilde{z}_{1}$ and $\tilde{u}_{1}$ oscillate between $+1$ and $-1$. This behavior may not be acceptable in practice and can be overcome by the hysteric quantizers studied in \cite{ceragioli.et.al.aut11} or the self-triggered controllers of \cite{depersis.frasca.necsys12}.

\begin{figure}[htpb]
\centering
\begin{minipage}{.5\textwidth}
\includegraphics [width=0.9\linewidth]{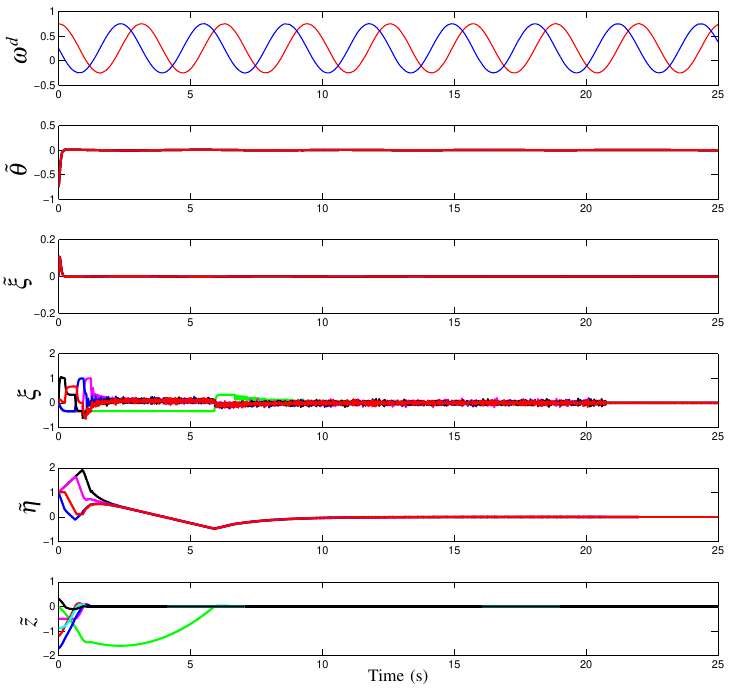}
\caption{The x-y components of the disturbance $w^d$ (chosen to be identical for all agents), together with the (horizontal) x-component of the disturbance estimation error $\tilde \theta$, the estimated velocity error $\tilde \xi$, the velocity error $\xi$, the reference velocity estimation error $\tilde \eta$ (for follower agents) and the horizontal component of the relative position vector $\tilde z$. Here, the disturbance is rejected by the controller designed in Section \ref{sec5}. As shown, $\tilde \theta$, $\tilde \xi$, $\xi$, $\tilde \eta$, $\tilde z$ converge to zero.}
\label{fig:f12}
\end{minipage}
\hfill
\begin{minipage}{.4\textwidth}
\centering
\includegraphics [width=0.9\linewidth]{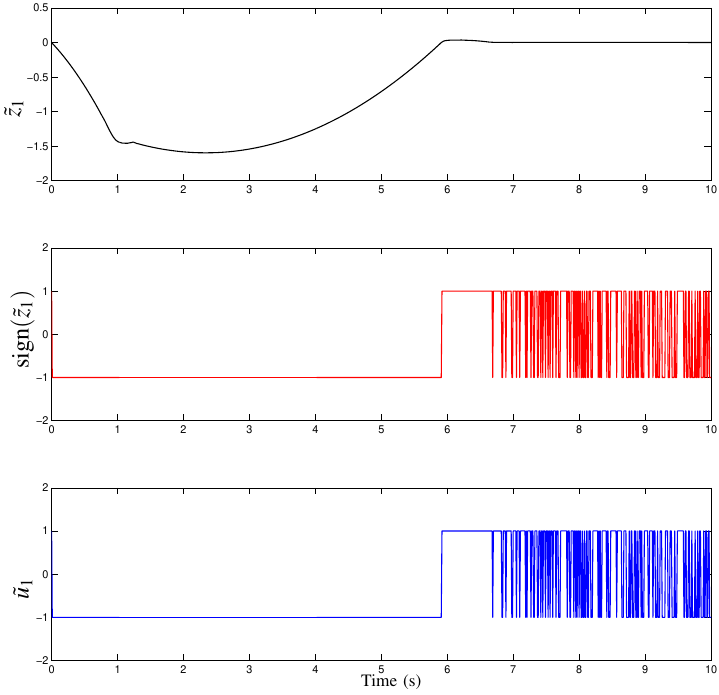}
\caption{The plot of the (horizontal) x-component of $\tilde{z}_{1}$, $\sign\; \tilde{z}_{1}$ and the corresponding control $\tilde{u}_{1}$ related to the simulation based on the controller proposed in  Section \ref{sec5}.}
\label{fig:f13}
\end{minipage}
\end{figure}

\section{Conclusion}
\label{sec:conclusion}
In this paper we considered a formation control problem with very coarse information for a network of strictly passive systems. We showed that despite the very coarse  information, the exact formation is reached. Moreover, the formation tracks a desired reference velocity even in the case when the reference velocity is only available to one of the agents (the so-called leader). 
In the same coarse sensing scenario and within the passivity framework, we designed internal-model-based controllers  for disturbance rejection and velocity tracking. 
\\ 
Possible future avenues of research include the extension of the results to deal with time-varying topologies. A few related results have been discussed in \cite{depersis.cdc11,xargay.et.al.cdc12}. Moreover, discontinuous control laws as those considered in this paper can be viewed as the outcome of  a non-smooth optimization problem associated with the original control problem (\cite{buerger.et.al.cdc11}) and it would be interesting to investigate this topic more in depth.  Another interesting topic to understand better is whether the finite-valued control laws of this paper can be used to tackle the case of (asymmetric) measurement noise. 
Finally, we observe that as the system converges to the prescribed formation, fast oscillations of the control inputs between $+1$ and $-1$ may occur. As discussed in the simulation section, these oscillations could be overcome by the hysteric quantizer of \cite{ceragioli.et.al.aut11} or the self-triggered approach of \cite{depersis.frasca.necsys12}. A comprehensive treatment of this aspect is another interesting topic that deserves attention. 


%

\appendix
\section {Non Smooth Control Theory Tools} 
{\it This appendix is added for the convenience of the Reviewers only. It will be omitted from the final version of the manuscript.}\\
We recall some basic notations from the theory of nonsmooth control which will be used throughout the paper. $x_{0} \in \R^{N}$ is a Krasowskii equilibrium for 
the differential inclusion $\dot x \in F(x(t))$ if the function $x(t) = x_{0}$ is a Krasowskii solution to $\dot x \in F(x(t))$ starting from the initial condition $x_{0}$, namely if ${\mathbf 0} \in {\mathcal K} (F(x))$. A set $\S$ is weakly (respectively, strongly) invariant for $\dot x \in F(x(t))$ if for any initial condition $\bar{x} \in \S$ at least one (all the) Krasowskii solution $x(t)$ starting form $\bar{x}$ belongs (belong) to $\S$ for all t in the domain of the definition of $x(t)$. Let $V$ be a locally Lipschitz continuous function. Recall that by Rademacher's theorem, the gradient $\nabla V$ of $V$ exists almost everywhere. Let $R$ be the set of measure zero where $\nabla V (x)$ does not exist. Then the Clarke generalized gradient of $V$ at $x$ is the set $\partial V(x) = co \{ \lim_{i \to +\infty} \nabla V(x_{i}) : x_{i} \rightarrow x, x_{i} \notin S, \;x_{i} \notin R\}$ where $S$ is any set of measure zero in $\R^{N}$. We define the set-valued derivative of $V$ at $x$ with respect to $(\ref{eq:b7})$ the set $\dot{\bar{V}}(x) = \{a \in \R : \exists v\in {\mathcal K} (F(x))\;\text{s.t.}\;a = p.v,\;\;\forall p \in \partial V(x)\}$.\\
{\bf Non-smooth LaSalle's invariance principle}. Let $V : \R^{n}\rightarrow \R$ be a locally Lipschitz and regular function. Let $\bar{x} \in \S$, with $\S$ compact and strongly invariant for $\dot x \in F(x(t))$. Assume that for all $\bar{x} \in \S$ either $\dot{\bar{V}} = \emptyset$ or $\dot{\bar{V}}\subseteq (-\infty,0]$. Then any Krasowskii solution to $\dot{x} \in F(x)$ starting from $\bar{x}$ converges to the largest weakly invariant subset contained in $\S \cap \{x \in \R^{n}: {\mathbf 0} \in \dot{\bar{V}}\}$, with $\mathbf 0$ the null vector in $\R^{n}$.
\end{document}